\documentclass[preprint,12pt,number]{elsarticle}
\usepackage{amsthm}
\usepackage{latexsym, amssymb, enumerate, amsmath}
\usepackage{fonts}
\usepackage{defs}
\usepackage[tight,footnotesize]{subfigure}
\newtheorem{theorem}{Theorem}
\newtheorem{definition}{Definition}
 \newtheorem{proposition}{Proposition}

\usepackage{algorithm}
\usepackage{algorithmic}%
\usepackage[normalem]{ulem}
\usepackage{slashbox}
\usepackage{afterpage}
\usepackage{multirow}
\usepackage{booktabs}
\usepackage[usenames,dvipsnames]{xcolor}

\usepackage[margin=1.0in]{geometry}

\newcommand{\ALGCOM}[1]{\COMMENT{\emph{#1}}}
\newcommand{\ALGCOMR}[1]{\hfill \COMMENT{\emph{#1}}}

\newcommand{\Q}[1]{\cQ\left( #1 \right)}

\newcommand{\tabref}[1]{Table~\ref{#1}}
\newcommand{\figref}[1]{Figure~\ref{#1}}
\newcommand{\algref}[1]{Algorithm~\ref{#1}}
\newcommand{\secref}[1]{Section~\ref{#1}}

\newcommand{\betain}[1]{\bsbeta_{#1,\rm{in}}}
\newcommand{\betaout}[1]{\bsbeta_{#1,\rm{out}}}

\DeclareMathOperator*{\argmin}{argmin}

\def\T{{\rm T}}

\def\etad{\eta_*}
\def\alphahat{\hat{\bsalpha}}
\def\alphabar{\bar{\bsalpha}}

\def\Gu{G_{\alphahat(\bu)}}
\def\abar{\bar{\bsalpha}}

\def\Rm{\cR_{\bm}}
\def\RQ{\Rm^{\cQ}}
\def\Rmone{\Rm |_{u_0=1}}
\def\RQone{\RQ |_{u_0=1}}

\def\co{{\rm co}}
\def\interior{{\rm int}}

\def\nQ{n_{\cQ}}

\def\tf{t_{\rm f}}

\def\xL{x_{\rm{L}}}
\def\xR{x_{\rm{R}}}

\def\gmax{\gamma_\text{max}}

\def\rseq{\{ r_\ell \}}

\begin{document}

\begin{frontmatter}

\title{Adaptive change of basis in entropy-based moment closures for linear
kinetic equations}
\date{\today}
\author[umd-ece]{Graham W. Alldredge \corref{cor} 
\fnref{rwth-address,other-support}}
\ead{alldredge@mathcces.rwth-aachen.de}
\address[umd-ece]{Department of Electrical and Computer Engineering
  \& Institute for Systems Research,
  University of Maryland
  College Park, MD 20742 USA
}
\cortext[cor]{Corresponding author.}
\fntext[rwth-address]{Present address: Center for Computational Engineering
Science,
RWTH-Aachen University,
Aachen, Germany}
\author[oak-ridge]{Cory D. Hauck\fnref{cory-support}}
\ead{hauckc@ornl.gov}
\address[oak-ridge]{Computational Mathematics Group,
				Computer Science and Mathematics Division,
				Oak Ridge National Laboratory,
				Oak Ridge, TN 37831 USA }
\author[umd-cs]{Dianne P. O'Leary\fnref{other-support}}
\ead{oleary@cs.umd.edu}
\address[umd-cs]{Department of Computer Science,
  University of Maryland
  College Park, MD 20742 USA}
\author[umd-ece]{Andr{\'e} L.~Tits\fnref{other-support}}
\ead{andre@umd.edu}
  
\fntext[cory-suppport]{This author's research was sponsored by the Office of
Advanced Scientific Computing Research and performed at the Oak Ridge National
Laboratory, which is managed by UT-Battelle, LLC under Contract No.
De-AC05-00OR22725.}
\fntext[other-support]{Supported by the U.S. Department of Energy, under Grant
DESC0001862.}


\begin{abstract}
Entropy-based ($\M_N$) moment closures for kinetic equations are defined by a
constrained optimization problem that must be solved at every point in a
space-time mesh, making it important to solve these optimization problems
accurately and efficiently. We present a complete and practical numerical
algorithm for solving the dual problem in one-dimensional, slab geometries. The
closure is only well-defined on the set of moments that are realizable from a
positive underlying distribution, and as the boundary of the realizable set is
approached, the dual problem becomes increasingly difficult to solve due to
ill-conditioning of the Hessian matrix.  To improve the condition number of the
Hessian, we advocate the use of a change of polynomial basis, defined using a
Cholesky factorization of the Hessian, that permits solution of problems nearer
to the boundary of the realizable set. We also advocate a fixed quadrature
scheme, rather than adaptive quadrature, since the latter introduces unnecessary
expense and changes the computationally realizable set as the quadrature
changes. For very ill-conditioned problems, we use regularization to make the
optimization algorithm robust. We design a manufactured solution and demonstrate
that the adaptive-basis optimization algorithm reduces the need for
regularization. This is important since we also show that regularization slows,
and even stalls, convergence of the numerical simulation when refining the
space-time mesh. We also simulate two well-known benchmark problems. There we
find that our adaptive-basis, fixed-quadrature algorithm uses less
regularization than alternatives, although differences in the resulting numerical
simulations are more sensitive to the regularization strategy than to the choice
of basis.

\end{abstract}


\end{frontmatter}


\section{Introduction}

Moment methods are commonly used to derive reduced models of kinetic transport.
Rather than fully resolve the kinetic distribution in phase space, moment models
instead track the evolution of a finite number of weighted velocity averages, or
moments of the distribution. Exact equations for these moments inevitably
require missing information about the unknown kinetic distribution that must be
approximated via a closure. Entropy-based closures approximate the full kinetic
distribution by an ansatz that solves a constrained, convex optimization
problem.   In the context of radiative transport \cite{Minerbo-1978,
Dubroca-Feugas-1999}, these models are commonly referred to as $\M_N$ (after
G.N.~Minerbo), where $N$ is the order of the highest-order moments of the model;
see \cite{AHT10} for additional references. These moment models preserve many
fundamental properties of the kinetic description, including positivity, entropy
dissipation, and hyperbolicity \cite{Levermore-1996}.

The primary drawback of the entropy-based approach is computational cost:
at least one optimization problem must be solved at every point on a 
space-time grid and, except for some $M_1$ cases, the optimization must be 
done numerically.  For some moments, the associated optimization problem can be
particularly expensive to solve.  These 
moments lie near the boundary of the \textit{realizable set}, defined as the set of vectors 
that are moments of a positive distribution. For realizable moments near the boundary,
the optimization algorithm may require a large number of iterations to converge (or may not converge at
all) and the solution will be sensitive to small changes in the moments. 
(Indeed, in some contexts, there even are realizable moments for which the
optimization problem has no
solution~\cite{Junk-1998,Junk-2000,Hauck-Levermore-Tits-2008,Schneider-2004}.)

It is most common to solve each optimization problem via
the associated convex dual.
For a smooth entropy function, a standard Newton method was proposed in
\cite{Mead-Papanicolaou-1984}.   
In this context, the difficulty in solving optimization problems for 
moments near the realizable boundary is characterized by an ill-conditioned
Hessian for the dual
objective function.  This matrix is a weighted integral in momentum 
space of a distribution of
rank-one matrices.  It becomes rank deficient (or nearly so) because the
weight function is an approximation of the underlying kinetic density, and for
moments near the realizable boundary, its mass will be concentrated around a
small number of directions in momentum space.  This is a common occurrence in
radiation applications. The sensitivity is further exacerbated by the necessity
of using an inexact quadrature and finite-precision arithmetic to approximate
the integrals. The contributions to the quadrature may effectively be zero for
most of the quadrature points, causing the computed Hessian to be singular.

The singularity in the dual Hessian has been addressed in different ways, 
in particular in the context
of some variations of Newton's method.
In \cite{Turek-1988}, the author takes
advantage of the structure of the Hessian, whose entries are themselves moments
of a known distribution that changes at each iteration.  Using orthogonal
polynomials with respect to this distribution (which are found using a standard
three-point recursion relation \cite[Ch. 22]{abramowitz-stegun-1964}), the
author is able to invert the ill-conditioned Hessian in a stable way. (See
\cite{Wheeler-1974} for an efficient algorithm to evaluate the recursion
coefficients.) More recently, orthogonal polynomials were used in the
multi-dimensional implementation found in \cite{Abramov-2009,Abramov-2009-2}.
There the author applies a BFGS quasi-Newton method (see, e.g.,~\cite[Ch.
6]{NW:06}) and, when the approximate Hessian becomes ill-conditioned, a
Gram-Schmidt procedure is applied to change into a polynomial basis for which
the Hessian is the identity.\footnote{In contrast to our problem,
ill-conditioning in \cite{Abramov-2009, Abramov-2009-2} occurs because the
moments are integrals over an unbounded domain,  so the Hessian is dominated by
moments corresponding to the highest order polynomials \cite{Junk-1998}.}
In \cite{Decarreau-Hilhorst-Lemarichal-Navaza-1992}, a penalized version of the
primal problem was introduced in order to handle nonrealizable moments.  This
modification amounts to Tikhonov regularization of the dual problem, which
also reduces ill-conditioning of the Hessian 
for realizable moments that are near the realizable boundary.
In \cite{Vie-Laurent-Massot-2012}, two modifications to the
Newton method for the solution of the dual problem
were introduced.  The first, which is only practical for a relatively
small number of moments, is to generate an initial guess for the Newton solver
by interpolating values from look-up tables. The second, which is only
practical in one-dimension, is a root finding algorithm to guide the
placement of nodes in the adaptive quadrature used to evaluate
the objective function and its derivatives. 
Finally, in \cite{AHT10}, where a damped Newton method is used, ill-conditioning
of the Hessian near the realizable boundary is avoided in two ways.
First, adaptive quadrature is used to better capture the support of the 
Hessian weight function and thereby increase the number of significant
rank-one contributions.
Second, a regularization method is introduced to move the moments away 
from the realizable boundary.
These two remedies are automatically invoked as needed, since
manual intervention is impractical.

Several specialized techniques for solving the dual problem have also been
suggested. In \cite{Borwein-Huang-1995, Huang-1996}, the authors 
show that, in the case of the Maxwell-Boltzmann entropy, the 
solution of the dual problem is the unique solution of a certain
finite set of linear equations. 
Unfortunately, setting up this linear system requires the 
knowledge of additional moments that are not available to
the closure. In \cite{Bandyopadhyay-Bhattacharya-Biswas-Parthapratim-Drabold-2005},
the authors solve the dual problem by means of a coordinate descent method, 
also known
as Bregman's balancing method \cite{Fang-Rajasekera-Tsao-1997,Bregman-1967},
where each sub-problem is solved using a multiple algebraic reconstruction
technique \cite{Gordon-Bender-Herman-1970,Fang-Rajasekera-Tsao-1997}.
Coordinate descent, however, is known to converge rather slowly
(see, e.g.,~\cite[p.230]{NW:06}).

In this paper, along the lines of~\cite{AHT10}, we employ a damped 
Newton method and investigate ways to better handle hard optimization 
problems near the realizable boundary.
We note that adaptive quadrature complicates realizability, so we do not use it.
Further, we show that regularization affects accuracy and therefore should 
be avoided
whenever possible. For this, we adopt the change-of-basis procedure introduced
in \cite{Abramov-2009,Abramov-2009-2}, albeit with a different implementation.
We focus on the Maxwell-Boltzmann entropy in the one-dimensional
setting on a bounded domain, although our methodology is applicable to general
smooth entropies and moments defined over bounded domains of arbitrary
dimension.

As in \cite{Abramov-2009, Abramov-2009-2}, we observe that the change of
basis makes the optimization more stable and effectively removes the need for an
adaptive quadrature.  Regularization is still required for robustness but is
invoked far less frequently, resulting in noticeable improvements in accuracy
in manufactured solution simulations we performed. 
We perform a series of numerical tests to quantify the overall performance of
the algorithm, to assess the interplay between the change of basis and the
regularization, and to determine efficient stopping criteria for the
optimization.  Our tests include a new manufactured solution and two well-known
benchmarks for transport in slab geometries.

The organization of the paper is as follows.  In Section \ref{sec:background},
we recall the one-dimensional kinetic equation, the derivation of entropy-based
moment models, and the issue of realizability that is central to the challenges
of implementation.  In Sections \ref{sec:alg} and
\ref{sec:implementation-details}, we present our ideas for solving the moment
closure problem numerically, including the adaptive change of basis,
the use of fixed quadrature, and the strategy for regularization. In Section
\ref{sec:results}, we give results of numerical experimentation, which analyze
the efficiency of different strategies combining the change of basis and the
regularization procedure.  We also examine the effects of several parameters
on algorithm efficiency.  Experiments include single optimization problems
which explore the realizable boundary, accuracy tests using manufactured
solutions, and two common benchmarks tests.  Finally, we draw conclusions in
Section \ref{sec:conclusions}.


\section{The Closure Problem}
\label{sec:background}
In this section, we provide a brief introduction to the closure problem,
following the detailed presentation in \cite{AHT10,Hauck-2011}.
We consider the migration of particles with unit
speed that are absorbed by or scattered isotropically off of a background
material medium with slab geometry.  The particle system is characterized by
a non-negative kinetic density $F=F(x,\mu,t)$ that is governed by a kinetic
transport equation
\begin{equation}
\p_t F + \mu \p_x F + \sig{t}F = \frac{\sig{s}}{2}\vint{F} \,,
\label{eq:transport}
\end{equation}
supplemented by appropriate boundary and initial conditions. The independent
variables in \eqref{eq:transport} are the scalar coordinate $x \in (\xL,\xR)$
along the direction perpendicular to the slab, the cosine $\mu\in [-1,1]$ of the
angle between the $x$-axis and the direction of particle travel, and time $t$.
Interactions with the material are characterized by non-negative variables
$\sig{s}(x)$, $\sig{a}(x)$, and $\sig{t}(x) := \sig{s}(x)+\sig{a}(x)$ which are
the scattering, absorption, and total cross-sections, respectively. For the
purposes of this paper, these cross-sections are assumed to be isotropic, i.e.,
independent of $\mu$. The angle brackets on the right-hand side of
\eqref{eq:transport} denote integration over $\mu$, i.e., for any integrable
function $g=g(\mu)$,
\begin{equation}
\vint{g} := \int^{1}_{-1} g(\mu) \, d\mu \:.
\end{equation}

Moment models for \eqref{eq:transport} are systems of partial differential
equations of the form
\begin{equation}
 \label{eq:moment_system_closed}
   \p_t \bu
   + \p_x  \bff(\bu)
   + \sig{t} \bu
   = \sig{s} Q \bu \:.
\end{equation}
Solutions $\bu= [u_0, u_1 \ldots, u_N]^\T : \bbR \times (0,\infty) \to
\bbR^{N+1}$ of \eqref{eq:moment_system_closed} provide an approximation to the
moments of $F$ with respect to linearly independent functions of $\mu$, i.e.,
$\bu(x,t) \simeq \vint{\bm F(x,\cdot,t)}$, where $\bm(\mu) = [m_0(\mu),
m_1(\mu), \ldots, m_N(\mu)]^\T.  $ While other choices are possible, we follow
standard practice \cite{Lewis-Miller-1984} and set $m_\ell$ to be the $\ell^{\rm
th}$ Legendre polynomial, normalized such that $\vint{m_\ell m_{\ell'}} = 2
\delta_{\ell,\ell'} /2(\ell+1)$. With this choice, the $(N+1) \times (N+1)$
matrix $Q$ is given by $Q_{\ell,\ell'}=\delta_{\ell,\ell'}\delta_{{\ell}, 0}$,
so that $Q\bu = [u_0, 0, \ldots, 0]^\T$.

The flux $\bff$ is determined by a closure.  For entropy-based models,
\begin{equation}
 \bff(\bu) := \Vint{\mu \bm \Gu}
\label{eq:flux}
\end{equation}
whenever $\bu$ is ``realizable'' (defined below),
where $\Gu$ is an ansatz for the underlying kinetic distribution
and solves the constrained, strictly convex optimization problem
\begin{equation}
  \minimize_{g \in L_+^{1}(d \mu)} ~ \Vint{\eta(g)} \qquad
  \mbox{subject to}~ \Vint{\bm g} = \bu \:.
\label{eq:primal}
\end{equation}
Here the kinetic entropy density $\eta \colon \bbR \to \bbR$ is strictly
convex and $L_+^{1}(d \mu)$ is the set of all non-negative, integrable
functions with respect to the Lebesgue measure $d \mu$.

The ansatz $\Gu$ belongs to a family of functions that are parameterized by
$\bsalpha \in \bbR^{N+1}$ and take the form $G_{\bsalpha}(\mu) =
\etad'(\bsalpha^\T \bm(\mu))$, where $\etad \colon \bbR \to \bbR$ is the
Legendre dual of $\eta$ and prime denotes differentiation. The Lagrange
multipliers $\alphahat(\bu)$ solve the unconstrained, strictly convex, dual
problem
\begin{equation}
 \alphahat(\bu) = \argmin_{\bsalpha \in \bbR^{N+1}} \left\{
  \Vint{\etad(\bsalpha^\T \bm)} - \bsalpha^\T \bu \right\}.
\label{eq:dual}
\end{equation}
(See \cite{Levermore-1996}  for more details.)  For the purposes of this paper,
we focus on the Maxwell-Boltzmann entropy $\eta(z) = z\log(z)-z$. Thus $\etad(y)
= \etad'(y) = e^{y}$ and
\begin{equation}
 G_{\bsalpha} = \exp(\bsalpha^\T\bm).
\label{eq:mbansatz}
\end{equation}

Problem~\eqref{eq:dual} does have a (unique) solution whenever $\bu$ is 
realizable, in the following sense.
\begin{definition}
A vector $\bv \in \bbR^{N+1}$ is said to be {\em realizable} (with respect to
$\bm$) if there exists a function $g \in L_+^1(d\mu)$ such that $\vint{\bm g} = \bv$.
The set of all realizable vectors is denoted by $\Rm$.
\end{definition}
The set $\Rm$ is an open, pointed, convex cone, and in the one-dimensional
setting is characterized by the positive-definiteness of Hankel matrices
\cite{Shohat-Tamarkin-1943}.  For the model problem considered here,
$\alphahat(\bu)$ is a diffeomorphism from $\Rm$ onto $\bbR^{N+1}$.  (See
\cite{Junk-1998,Borwein-Lewis-1991,Mead-Papanicolaou-1984}.)  Moments on the
boundary of realizability $\p \Rm$ are uniquely realized by atomic
measures---i.e., on the boundary of realizability, the kinetic distribution is a
sum of delta functions \cite{Curto-Fialkow-1991}.

A numerical method for solving~(\ref{eq:moment_system_closed}) must preserve $\Rm$.  To
this end, a finite-volume kinetic scheme was introduced in \cite{AHT10}, which
takes the semi-discrete form
\begin{equation}
 \p_t \bu_j + \frac{\Vint{\mu \bm G_{j + 1/2}} - \Vint{\mu \bm G_{j - 1/2}}}{\dx} + \sig{t} \bu_j =
  \sig{s} Q \bu_j \:,
  \label{eq:semidisc}
  \end{equation}
where $\bu_j$, for $j \in \{ 1, \ldots , N_x\}$, approximates the cell average
$\bu(x,t)$ over an interval $I_j =(x_{j-1/2},x_{j+1/2}) \subset (\xL,\xR)$ and
$G_{j\pm1/2}$ is an approximation of the entropy ansatz at the cell edge
$x_{j\pm1/2}$ based on a linear reconstruction of $G_j = G_{\alphahat(\bu_j)}$
and a standard minmod-type limiter. Time integration is performed using the
second-order strong-stability-preserving Runge-Kutta (SSP-RK2) method
\cite{Gottlieb-Shu-Tadmor-2001}, also known as Heun's method or the improved
Euler method. This is a two-stage method and thus requires the dual problem
\eqref{eq:dual} to be solved twice for every unknown in space and time. SSP
integrators are used because, under appropriate conditions, they preserve convex
sets.  We let $\bu^n_j$ denote the numerical solution at time step $n$ in cell
$j$ for $n \in \{0, \ldots , N_t\}$. The boundary conditions are implemented by
prescribing realizable moments in  ghost cells indexed by $j \in \{-1, 0, N_x +
1, N_x + 2 \}$ at each stage of the Runge Kutta method.


\section{Basics of the Optimization}
\label{sec:alg}

We focus in this section on 
components for
efficiently solving
the dual problem \eqref{eq:dual}.  
Because the objective is smooth, unconstrained, and strictly convex, we use
Newton's method, stabilized by an Armijo backtracking line search \cite{armijo}.
Our optimization algorithm computes an approximation
$\alphabar$ to the true solution $\alphahat$.
The dual objective function $f:\bbR^{N+1} \to \bbR$ is
\begin{equation}
 f(\bsalpha) := \Vint{G_{\bsalpha}} - \bsalpha^\T \bu \,.
\label{eq:f}
\end{equation}
Its gradient $\bg:\bbR^{N+1} \to \bbR^{N+1}$ and Hessian $H:\bbR^{N+1} \to
\bbR^{(N+1)\times  (N+1)}$
are given by (recall~(\ref{eq:mbansatz}))
\begin{equation}
 \bg(\bsalpha) := \Vint{\bm G_{\bsalpha} } - \bu 
 \quand
 H(\bsalpha) := \Vint{\bm \bm^\T G_{\bsalpha}} \:,
 \label{eq:grad}
\end{equation}
and the Newton direction $\bd(\bsalpha)$ solves the linear system $ H(\bsalpha)\bd(\bsalpha) = -\bg(\bsalpha)$.

Our optimization algorithm has four important components:  an adaptive change of
basis to improve the conditioning of the Hessian, appropriate stopping criteria,
a fixed quadrature set for approximating integrals, and a regularization method
used for very ill-conditioned problems. We now consider each of these
components.

\subsection{Adaptive change of basis}
\label{subsec:change-of-basis}

Following \cite{Abramov-2009}, we apply a change of basis to improve the
condition number of the Hessian.  Specifically, when expressed in the new basis,
the Hessian at the current iterate becomes the identity matrix.  In
\cite{Abramov-2009}, a BFGS algorithm is used, and the change of basis is
invoked only when the condition number of the approximate Hessian is greater
than a certain threshold.  Here we use a damped Newton method, and we invoke
such change of basis at every iteration.

At iteration $k$, let $S_{k}$ be an invertible matrix which determines a
new polynomial basis $\bp_{k} = S_{k} \bm$ and let $T_{k} = S_{k}^{{-1}}$.  If
$\bsalpha_{k}$ is the dual variable at iteration $k$ with respect to basis 
$\bm$, then
let $\betain{k} = T^{\T}_{k-1}\bsalpha_{k}$ be the dual variable at iteration
$k$ with respect to 
basis $\bp_{k-1}$ 
and $\betaout{k} = T^{\T}_{k}\bsalpha_{k}$ be the dual variable at
iteration $k$ after changing to the new basis $\bp_{k}$.

Define a new objective $f_{k}:\bbR^{N+1}\to\bbR$ by 
\begin{equation}
 f_{k}(\bsbeta) := f(S_k^\T \bsbeta)
  = \vint{\exp(\bsbeta^\T S_{k} \bm)} - \bsbeta^\T S_{k} \bu,
\end{equation}
so that $f(\bsalpha) = f_{k}(T_{k}^\T\bsalpha)$ for all $\bsalpha$.
Then $f_{k}$ is strictly convex with gradient and (positive-definite) Hessian
\begin{equation}
 \bg_k(\bsbeta) 
  = S_{k} \bg (S_{k}^{\T}\bsbeta)
  \quand
   H_{k}(\bsbeta) 
  = S_{k} H(S_{k}^{\T}\bsbeta) S_{k}^{\T} \,.
\label{eq:gradHM}
\end{equation}
The Newton step $\bd_{k}$ for  $f_{k}$ at $\bsbeta$ solves $H_{k}(\bsbeta)
\bd_{k}(\bsbeta) =  -\bg_k(\bsbeta)$.

Clearly $H_{k}(\betaout{k})=I$ if and only if $T_{k}$ factors $H(\bsalpha_{k})$,
i.e., $H(\bsalpha_{k}) \equiv H(S_{k}^{\T}\betaout{k}) = T_{k}T_{k}^{\T}$, in
which case the Newton direction with respect to $\bp_k$ coincides with the
steepest
descent direction:
\begin{equation}
\label{eq:dMk}
 \bd_{{k}}(\betaout{k}) = -\bg_{{k}}(\betaout{k}) 
  = \bu_{k} - \Vint{\bp_k \exp(\betaout{k}^{\T} \bp_k)} \,,
\end{equation}
where $\bu_{k} = S_k \bu = \Vint{\bp_k \Gu}$ is the moment vector expressed
in the $\bp_k$ basis.
Furthermore, $\bp_k$ is orthonormal with respect to the weight 
$G_{\bsalpha_k}=\exp(\bsalpha_k^\T\bm)=\exp(\betaout{k}^{\T} \bp_k)$,
since
\begin{equation}
\vint{\bp_k \bp_k^\T \exp(\bsalpha_k^{ \T} \bm)}
= S_k \vint{\bm \bm^\T \exp(\bsalpha_k^{ \T} \bm)} S^{\T}_k
= S_k H(\bsalpha_k) S^{\T}_k
= H_k(\betaout{k}) = I.
\end{equation}

If $H(\bsalpha_{k})$ is ill-conditioned, then a direct
computation and application of $T_{k}$ may be inaccurate. It is more stable to
change bases iteratively. 
To this end, let $L_{k}$ be any matrix such that
\begin{equation}
H_{k-1}(\betain{k}) = L_{k} L_{k}^\T.
\label{eq:H=LL'}
\end{equation}  Using this formula, it is a simple
exercise to show that $T_{k} = T_{k-1}L_{k}$ factors $H(\bsalpha_{k})$,
that $\betaout{k} = L_k^{T} \betain{k}$, and that 
$\bu_{k} = L_k^{-1} \bu_{k-1}$.

In exact arithmetic, this change of basis has no effect on the sequence of
Newton iterates $\bsalpha_k$.  Using inexact arithmetic, however, we 
observed (see section~\ref{sec:results} below)
that when $\bu$ is 
near the realizability boundary, under the proposed change of basis,
the stability of the iteration is greatly
improved: the Hessian matrix $H(\bsalpha_k)$ in the original basis
is highly ill-conditioned, so performing matrix computations with $L_k$ and
$T_k$---whose condition numbers are the square root of those of $H_{k -
1}(\betain{k})$ and $H(\bsalpha_k)$ respectively---instead of with
$H(\bsalpha_k)$ should reduce errors. Furthermore, when the Hessian matrix in
the original basis is poorly conditioned, the computed Newton direction may
even fail to be a direction of descent for the objective function.  
In contrast, in the new coordinate system, the step is taken
in the direction of the negative of the computed gradient.  Even in inexact
arithmetic this computed step is quite likely to have a negative inner product
with the true gradient and thus be a descent direction. 

\subsection{Stopping the Newton iteration}
\label{subsec:stopping}

Following \cite{AHT10}, our stopping criterion involves two conditions:
\begin{equation}
 \|\bg(\bsalpha_k)\|_{2} \leq \tau  \quand
 \exp\left(5\|\bd(\bsalpha_k)\|_1\right) \leq 1+\veps_\gamma \:.
\label{eq:stopping}
\end{equation}
In view of \eqref{eq:grad}, the first condition bounds the Euclidean distance
between $\bu$ and the moments of the candidate ansatz
$G_{\bsalpha_k}$. Moreover, because the spectral radius of the Jacobian of
$\bff$ is bounded
by one \cite{Olbrant-Hauck-Frank-2012}, it also bounds the error in the flux
$\bff$ (see \eqref{eq:flux}, \eqref{eq:mbansatz}):
\begin{equation}
\| \bff(\bu) - \Vint{\mu \bm G_{\bsalpha_k}} \|_2
\leq \sup_{\bv \in \Rm} \left \| \frac{\p \bff}{\p \bv} (\bv) \right \|_2
 \| \bu - \Vint{ \bm G_{\bsalpha_k}}  \|_2 \le \|\bg(\bsalpha_k)\|_{2}
\end{equation}

The second condition estimates an upper bound on 
\begin{equation}
\gamma(\mu) := G_{\bsalpha_k}/{\Gu}
= G_{\bsalpha_k - \alphahat(\bu)},
\label{eq:gamma}
\end{equation}
the ratio of the ansatz associated with a current iterate $\bsalpha_k$ to the ansatz
of the solution $\alphahat(\bu)$. The purpose of this condition (see Theorem 
\ref{thm:invariance} below) is to maintain realizability of the moments
generated by the kinetic scheme in Section \ref{sec:background}. Following
\cite{AHT10}, we use the Newton direction $\bd(\bsalpha_k)$ to approximate
$\bsalpha_k - \alphahat(\bu)$, but rather than the two-norm estimate used in
\cite{AHT10}, we use a tighter estimate to bound $G_{\bd(\bsalpha_k)}$:
\begin{equation}
 \max_{\mu \in [-1,1]} G_{\bd(\bsalpha_k)} 
 = \max_{\mu \in [-1,1]} \exp \left( \bd(\bsalpha_k)^\T \bm \right)
  \leq \exp\left(\|\bd(\bsalpha_k)\|_1\right)\, ,
\label{eq:L1bnd}
\end{equation}
where we have used the fact that 
$\max_{\mu} |m_i| = 1$ for all $i$. Adding a safety factor of $5$ gives
\eqref{eq:stopping}, which ensures with high confidence that
\begin{equation}
\label{eq:gamma<1+eps}
\gamma(\mu) \leq 1+\veps_\gamma \:.
\end{equation}

\subsection{Fixed Curtis-Clenshaw quadrature}
\label{subsec:quad}

The integrals in the objective function (cf.~\eqref{eq:f}) and its gradient and
Hessian (cf.~\eqref{eq:grad}) cannot, in general, be computed explicitly.
Therefore a numerical quadrature rule must be used. Let $\cQ$ be a quadrature
rule defined for functions $g:[-1,1] \to \bbR$ by
\begin{equation}
 \cQ(g) = \sum_{i = 1}^{\nQ} w_i g(\mu_i),
\end{equation}
where the quadrature nodes $\{\mu_i\}_{i = 1}^{\nQ}$ and the quadrature weights
$\{w_i\}_{i = 1}^{\nQ}$ are chosen so that $\cQ(g)$ approximates $\Vint{g}$.
For numerical computations, $\Vint{\cdot}$ should always be understood as $\Q{\cdot}$.  The specific meaning should be clear from the context.
We define the {\em $\cQ$-realizable set}
\begin{equation}
\label{eq:RQ}
 \RQ := \left\{ \bu \,\left|\, \bu = \sum_{i = 1}^{\nQ} w_i \bm(\mu_i) f_i,\,
  f_i > 0  \right. \right\}  \:.
\end{equation}
Note that $\RQ$ is a {\em strict} (polytopic) subset of $\Rm$ and that, like
$\Rm$, it is an open, pointed, convex cone.  In particular \begin{equation}
\label{eq:RQgen}
 \RQ = \left\{ \bu \,\left|\, \bu = c\bv\,, c > 0\,, \bv \in \RQ|_{u_0=1}
\right. \right\}\,.
\end{equation}

In \cite{AHT10}, an adaptive quadrature was used to reduce the condition number
of the Hessian.  However, the use of an adaptive quadrature
introduces serious numerical difficulties. For example, 
the $Q$-realizable set changes with the choice of quadrature
nodes.  Thus an iterate that is realizable can suddenly become non-realizable when the
quadrature changes, and this forces the use of artificial techniques like
regularization, discussed below, in order to continue the computation.  In
contrast, a fixed quadrature $\cQ$ makes it easy to keep the numerical solution
within the $\cQ$-realizable set.

\begin{theorem}
Let $\bu_j^{n+1}$, $j\in \{1,\ldots,N_x\}$, be defined via the kinetic
scheme in Section \ref{sec:background}, with time-step restriction
\begin{equation}
 \gmax \frac{\dt}{\dx}\frac{\theta + 2}{2} + \sig{t}\dt < 1 \:,
\label{eq:cfl}
\end{equation}
and let $\gmax$ be  the maximum value of $\gamma(\mu_i)$ (cf. \eqref{eq:gamma})
over all quadrature nodes, spatial cells, and stages of the Runge-Kutta method
used to integrate the kinetic scheme in time.  If $\bu^n_j \in \RQ$ for $j\in
\{-1,\ldots,N_x+2\}$ and if the moments in the ghost cells are in $\RQ$ at each
stage of the Runge-Kutta scheme, then $\bu^{n+1}_j \in \RQ$ for $j \in
\{1,\ldots,N_x\}$.
\label{thm:invariance}
\end{theorem}

\noindent The proof of Theorem \ref{thm:invariance} is a trivial modification of
the proof of Theorem 2.5 in \cite{AHT10}.

It has also been observed
that the size of the error in many adaptive quadrature rules does not decrease
monotonically with the number of points and, furthermore, that the number of
points required to satisfy a tolerance criterion is often much larger than the
number of points needed for an accurate evaluation of the integral
\cite{Lyness-1983}.  These issues lead to a considerable increase in the
complexity and computational time of the optimization algorithm.  
Moreover, as shown in \cite{AHT10},
refining the quadrature does not help if the exact Hessian is ill-conditioned.
With the adaptive basis, the
condition number of the Hessian is kept under control by iteratively changing
the polynomial basis.  Thus we opt to use a fixed quadrature and avoid the above
numerical complications.

Although the dual problem \eqref{eq:dual} has a solution for all $\bu \in \Rm$,
the use of a quadrature-based {\em approximation} of the dual objective function
means there will be a solution if and
only if $\bu \in \RQ$.  Consequently, it is important to choose quadratures $Q$
for which $\Rm \backslash \RQ$ is small.  In light of \eqref{eq:RQgen}, the
following characterization of $\RQone$ helps guide this choice.
\begin{proposition}
For any quadrature $\cQ$ using positive weights $w_i$,
\begin{equation}
 \left. \RQ \right|_{u_0=1} = \interior\, \co \{\bm(\mu_i) \}_{i = 1}^{\nQ}
\,,
\end{equation}
where $\co$ indicates the convex hull and $\interior$ the interior.
\end{proposition}

\begin{proof}
Let $\bu \in \RQone$.  Then, from~\eqref{eq:RQ}, $\bu=\sum\lambda_i\bm(\mu_i)$,
with $\lambda_i := w_i f_i > 0$.  Also, $\sum \lambda_i = 1$ since $u_0=1$ and
$m_0 \equiv 1$, and therefore $\bu \in \interior \, \co \{\bm(\mu_i) \}$. On the
other
hand, if $\bu \in \interior \, \co \{\bm(\mu_i) \}$, there must exist scalars
$\lambda_i$ such that $\bu=\sum\lambda_i\bm(\mu_i)$, with $\lambda_i>0$ and
$\sum \lambda_i = 1$. Choosing $f_i := \lambda_i / w_i$ shows that $\bu \in
\RQone$.
\end{proof}

\begin{figure}
 \centering
 \subfigure[Using four-point Gauss-Legendre quadratures on $[-1,0{]}$ and
  ${[}0,1{]}$.]{ 
 \label{subfig:RmQgl}
  \includegraphics[scale=0.36,viewport=130 200 520 500]{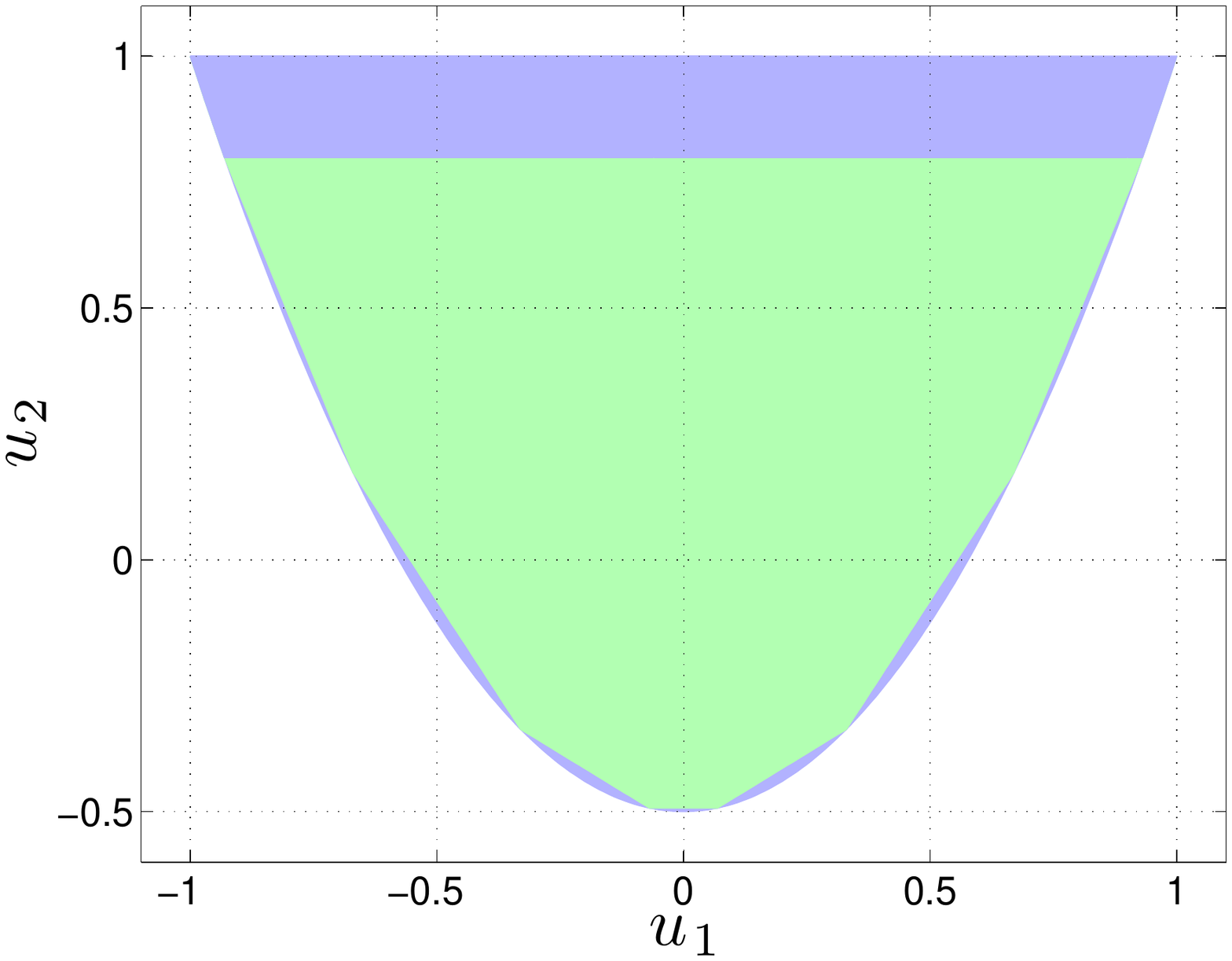}}
 \hspace{1.3cm}
 \subfigure[Using three-point Curtis-Clenshaw quadratures on $[-1,0{]}$ and ${[}0,1{]}$.] { \label{subfig:RmQcc}
  \includegraphics[scale=0.36,viewport=130 200 520 500]{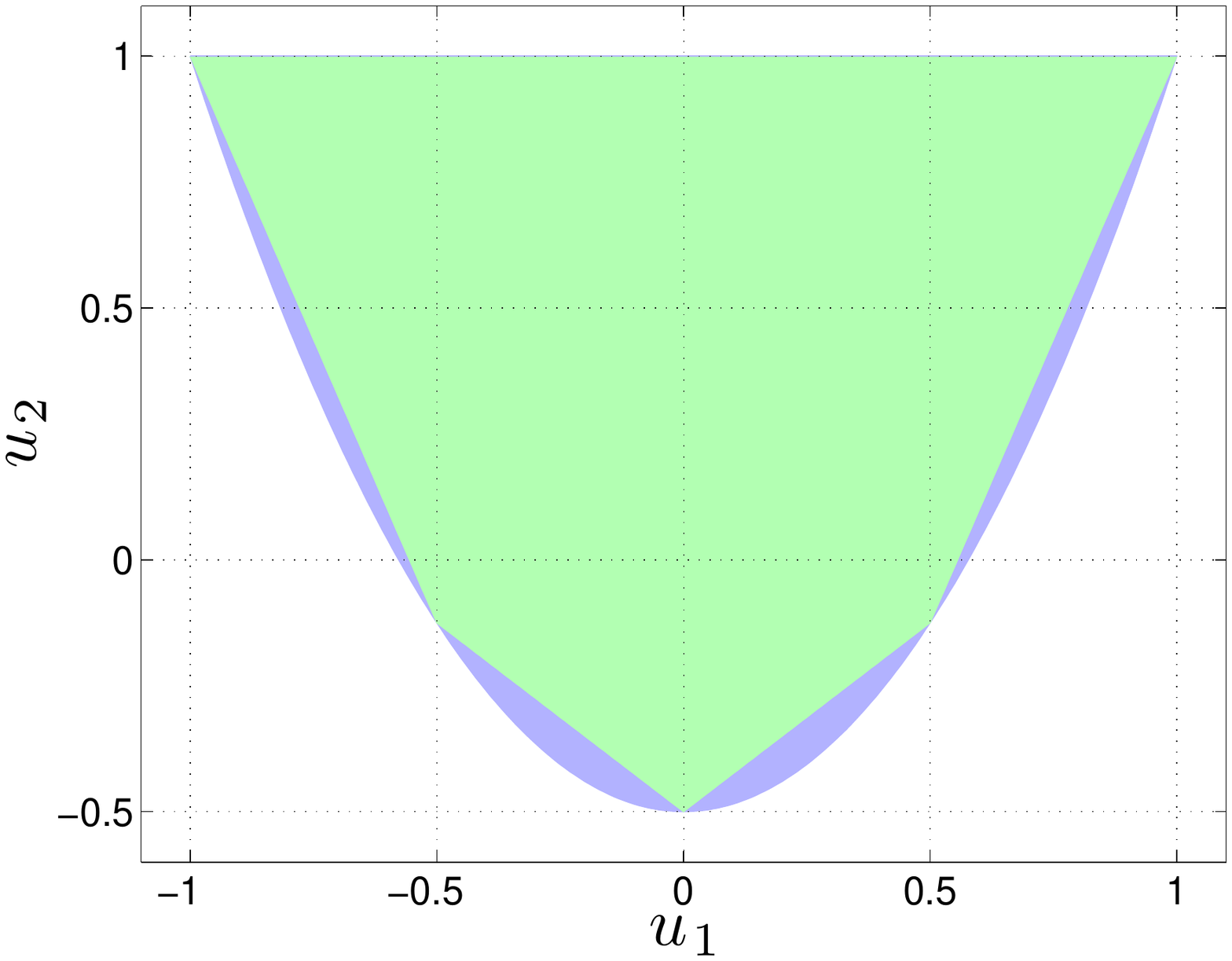}}
 \caption{Illustrating $\RQ|_{u_0=1}$ for $\M_2$.  The green indicates
  $\RQone$, and the blue indicates $\Rmone \backslash \RQone$.
  Ideally, the two sets would coincide. }
 \label{fig:RmQ}
\end{figure}

It remains to select a quadrature rule. \figref{fig:RmQ} shows examples of
$\RQ$ in the $\M_2$ case ($\bm = [1, \mu, \frac12
(3\mu^2 - 1)]^T$) using low-order Curtis-Clenshaw and Gauss-Legendre 
quadrature rules. For this choice of $\bm$,
$\Rmone$ is flat on the top, but curved at the bottom. 
If the endpoints $\mu = \pm 1$ are nodes of $\cQ$, as in the case of
Curtis-Clenshaw quadrature, the entire flat portion at the top is in
$\partial\RQone$ (cf. \figref{subfig:RmQcc}). The Gauss-Legendre
quadrature, on the other hand, does not contain the endpoints and, as a
consequence, leaves large regions of $\Rm$ out of $\RQ$ (cf.
\figref{subfig:RmQgl}). With Curtis-Clenshaw, $\RQ$ contains all realizable
moments for which $|u_1|/u_0$ is arbitrarily close to one. Such moments occur in
many situations, including the plane source benchmark problem simulated in
Section \ref{sec:results}. These observations motivate our use of
Curtis-Clenshaw quadrature rather than Gauss-Legendre quadrature which was used,
for example, in~\cite{AHT10,Mead-Papanicolaou-1984}.

\subsection{Regularization}
  
Even with the adaptive basis, there are realizable moments for which the damped
Newton method does not converge in a reasonable number of iterations.  
Indeed $H_{k-1}(\betain{k})$ can be highly ill-conditioned even though in the previous step
$H_{k-1}(\betaout{k-1})=I$.  In extreme situations, this may cause the factorization in \eqref{eq:H=LL'}
to fail because of round-off errors or to be so inaccurate that
the next Newton step is not a descent direction.  In the latter case, it is possible to 
refactor $H_{k-1}(\betain{k})$ (see Section \ref{subsec:chol-basis}) in order to find a descent 
direction.  However, refactorization may be required many times.

To address this issue, we
employ the regularization scheme introduced in \cite{AHT10}. The regularization
strategy is simple: moments $\bu$ for which the dual problem \eqref{eq:dual} is
deemed  too difficult to solve (by some prescribed criteria) are replaced by
nearby moments $\bv(\bu,r)$ that are further away from the boundary.  These
nearby moments are computed by taking the convex combination of $\bu$ with the
moments of the isotropic distribution with the same particle density:
\begin{equation} \bv(\bu,r)
:= (1 - r)\bu + rQ \bu \label{eq:v(r)}
\end{equation}
where $0 < r \ll 1$.  When $\bu$ is near $\p \Rm$, $\alphahat(\bv(\bu,r))$ is
typically much easier to compute than $\alphahat(\bu)$, even for small values of
$r$.  

The regularization procedure does introduce errors of order $r$ in the numerical
solution.  However, with adaptive change of basis, the need for regularization
is reduced relative to the fixed-basis method (which was used in \cite{AHT10}).


\section{Implementation Issues}
\label{sec:implementation-details}

Now that we have a broad outline for our optimization algorithm we discuss the
implementation details. The complete algorithm is presented in \algref{alg:r}.

\begin{algorithm}
\caption{The optimization algorithm with regularization.}
\label{alg:r}
\begin{algorithmic}
 \REQUIRE $\bu \in \Rm \subset \bbR^{N+1},\: \betain{0} \in
  \bbR^{N+1},\: T_{-1} \in \bbR^{(N+1) \times (N+1)}$
 \ALGCOM{$\bu$ is assumed to be in the Legendre basis $\bm$, while
  $\betain{0}$ is assumed to be in the $\bp_{-1} = T_{-1}^{-1}\bm$ basis}
 \newline \hspace*{-0.215in} \textbf{Parameters:}
  $\tau > 0$,
  $\veps_\gamma > 0$,
  $k_0 \in \bbN$,
  $\xi \in (0, 1/2)$,
  $\chi \in (0, 1)$,
  $k_{\max} \in \bbN$,
  $\rseq \subset [0, 1]$, an increasing sequence starting at zero,
  $\veps > 0$,
  $\cQ$, a quadrature rule.
 \STATE $r_{\max} \gets \max \rseq$
 \FOR{$r \in \rseq$}
  \STATE $P_{-1} \gets T_{-1}^{-1}\bm$
  \STATE $\bv \gets ((1 - r)\bu + r Q \bu)$
  \STATE $\bv_{-1} \gets T_{-1}^{-1}\bv$
  \STATE $f_0 \gets \Q{\exp\left(\betain{0}^{\T} \bp_{-1} \right)}
   - \betain{0}^{\T} \bv_{-1}$
  \FOR{$k \in \{0, 1, 2, \ldots , k_{\max}\}$}
   \STATE $[{\tt chol\_flag}, \betaout{k},\, \bv_k,\, \bg_k,\, P_k,\, T_k]
    \gets {\tt change\_basis}(\betain{k}, \bv_{k-1}, P_{k-1},
    T_{k-1})$
   \IF{ {\tt chol\_flag} $=$ false and
    $r = r_{\max}$}
    \RETURN failure to converge
   \ENDIF
   \IF{($k > k_0$ or
    {\tt chol\_flag} $=$ false) and
    $r < r_{\max}$}
    \STATE \ALGCOM{Exit the inner \emph{for} loop so that
     $r$ is increased.}
    \STATE \textbf{break for} 
   \ELSE
    \STATE $e_k \gets \|\Q{\bm \exp(\betaout{k}^{\T} T_k^{-1} \bm)}
     - \bv\|$
    \STATE $\bd_k \gets -\bg_k $
    \IF{$e_k < \tau$ and $\exp(5\|T_k^{-\T}\bd_k\|_1) < 1 +
     \veps_{\gamma}$}
     \STATE $\abar \gets T_k^{-\T} \betaout{k}$
     \RETURN $\abar,\:  T_k$
    \ELSE
     \STATE $\zeta_k \gets 1$
     \STATE $\betain{k+1} \gets \betaout{k}$
      \WHILE{$\zeta_k > \veps \| \betaout{k} \| /\| \bd_k \|$}
       \STATE $f \gets \Q{\exp\left(\left(\betaout{k} + \zeta_k
        \bd_k\right)^{\T} \bp_k \right)} - \left(\betaout{k} + \zeta_k
        \bd_k\right)^{\T} \bv_k$
       \IF{$f \le f_k + \xi\zeta_k \bg_k^T\bd_k$}
        \STATE $\betain{k+1} \gets \betaout{k} + \zeta_k \bd_k$
        \STATE $f_{k + 1} \gets f$
        \STATE \textbf{break while}
       \ENDIF
       \STATE $\zeta_k \gets \chi \zeta_k$
      \ENDWHILE
    \ENDIF
   \ENDIF
  \ENDFOR
 \ENDFOR
 \STATE \RETURN failure to converge
\end{algorithmic}
\end{algorithm}

\subsection{Defining the orthonormal basis}
\label{subsec:chol-basis}
There are many different orthogonal polynomial bases with respect to
$G_{\bsalpha} d \mu$.
In \cite{Abramov-2009}, orthogonal bases are computed
in which each basis polynomial
has the same degree.  Indeed, because in \cite{Abramov-2009} the
integration domain is unbounded, components of the dual variable associated
with basis polynomials of higher polynomial degrees are much more sensitive
than those associated to lower degrees.
In contrast, the domain of $\mu$ in the current application is bounded; hence a
triangular basis $\bp_k = [p_{k,0},p_{k,1},\ldots p_{k,N}]^T$, where the polynomial 
$p_{k,\ell}$ has degree $\ell$, does not present
the same numerical difficulties.  Rather, in this case, such a basis is
preferable since it leads to simpler matrix operations. In addition, since
${p_{k,0}}$ is a constant, orthogonality of the basis $\bp_k$ with respect to
$G_{\bsalpha_k}$ implies that
$\vint{p_{k,\ell} G_{\bsalpha_k}}=0$ for $\ell >0$.  This simplifies the computation of
the gradient $\bg_k$ (see \eqref{eq:dMk}).

The simplest way to maintain a triangular basis is to 
let $L_k$ in~(\ref{eq:H=LL'}), which defines the iterative change
of basis, be the Cholesky factor of $H_{k-1}(\betain{k})$, i.e.,
be lower triangular and positive-definite.  
The steps to orthonormalize the
basis using the Cholesky factor
are given in Algorithm
\ref{alg:change-of-basis}, which at each iteration computes the new multipliers
$\betaout{k}$, moments $\bu_k$, gradient $\bg_k$ of the dual objective function,
the  $(N+1) \times \nQ$ matrix $P_k$ of values of the basis polynomials at the
quadrature nodes, and the change of basis $T_k$. Once a triangular basis is initialized, the successive bases $\bp_k$ remain
triangular, and an initial triangular orthogonal basis is easily available:  If
we start at the multipliers $\bsalpha_{\rm iso}=(\log(u_0/2), 0, \ldots ,
0)^{\T}$ associated with the isotropic distribution, then the Legendre
basis~$\bm$ (which indeed is triangular) is a natural choice because it is
orthogonal with respect to the isotropic ansatz $G_{\bsalpha_{\rm iso}}$.

\begin{algorithm}[tb]
\caption{The {\tt change\_basis} steps
used to produce an orthonormal basis using the Cholesky factorization.}
\label{alg:change-of-basis}
\begin{algorithmic}
 \REQUIRE $\bsbeta_{\rm in} \in \bbR^{N+1},\:
          \bu_{\rm in} \in \bbR^{N+1},\:
          P_{\rm in} \in \bbR^{(N + 1) \times \nQ},\:
          T_{\rm in} \in \bbR^{(N + 1) \times (N + 1)}$

 \STATE \ALGCOM{The Hessian is initially in the $\bp_{\rm in}$
  basis,\footnotemark and values of these polynomials at the quadrature nodes
  are stored in $P_{\rm in}$.}
 \STATE $H \gets \Q{\bp_{\rm in} \bp_{\rm in}^\T
  \exp(\bsbeta_{\rm in}^\T \bp_{\rm in})}$.
 \STATE $(L, {\tt chol\_flag}) \gets {\tt chol}(H)$
  \ALGCOMR{{\tt chol\_flag} is \emph{false} if the Cholesky factorization
   fails}
 \IF{{\tt chol\_flag} = false}
  \RETURN {\tt chol\_flag}, 0, 0, 0, 0, 0
 \ENDIF
 \STATE $P_{\rm out} \gets L^{-1} P_{\rm in}$
 \STATE $T_{\rm out} \gets T_{\rm in} L$ 
 \STATE $\bsbeta_{\rm out} \gets L^\T \bsbeta_{\rm in}$
 \STATE $\bu_{\rm out} \gets L^{-1} \bu_{\rm in}$
 \STATE $\bg_{\rm out} \gets (p_{0,{\rm out}}\Q{\exp\left( \bsbeta_{\rm out}^\T
  \bp_{\rm out} \right)},0,\ldots,0)^\T - \bu_{\rm out}$
 \RETURN {\tt chol\_flag},\: $\bsbeta_{\rm out},\: \bu_{\rm out},\: \bg_{\rm
  out},\: P_{\rm out},\: T_{\rm out}$
\end{algorithmic}
\end{algorithm}
\footnotetext{For clarity, we use $\bp_{\rm in}$ to refer to the basis whose evaluations are
stored in $P_{\rm in}$. Actual calculations are performed using $P_{\rm in}$.}

In exact arithmetic and when applied to the same original basis, Cholesky and
modified Gram-Schmidt (used in~\cite{Abramov-2009}) both yield the same 
new basis, up to multiplication of individual basis polynomials by $\pm 1$.
We chose the Cholesky method because it is less computationally expensive than
Gram-Schmidt and, in fact, would be recommended in computing the Newton step
even
if no change of basis was performed. Considering only the highest-order terms,
the Cholesky method uses $\nQ N^2 / 2$ multiplications to form the Hessian, $N^3
/ 6$ multiplications to factor the Hessian, and $\nQ N^2 / 2$ multiplications
to update the array storing the evaluation of the basis polynomials at the
quadrature nodes.  In contrast, the
modified Gram-Schmidt method does not form the Hessian but instead
requires $\nQ N^2$ multiplications to evaluate the necessary inner products 
and $\nQ N^2 / 2$ multiplications to update the array storing the
evaluation of the basis polynomials at the quadrature nodes.   For the numerically
computed Hessian to have full rank, it is necessary that $\nQ \geq N+1$.  The benefit 
of using Cholesky increases as $\nQ$ increases.

While the Gram-Schmidt algorithm is somewhat more stable numerically, in our
experience, the difference is negligible, partly because when the Cholesky
computation is inaccurate, our algorithm (Algorithm~\ref{alg:r}) 
automatically reorthogonalizes by recomputing the Cholesky factor.  
Indeed, suppose 
the line search fails, that is, at some iteration $k$, it backtracks all 
the way to $\betaout{k}$.  Since we did not compute the 
Hessian $H_k(\betaout{k})=L_k^{-1}H_{k-1}(\betaout{k})L_k^{-\T}$ 
but rather assumed it was the identity,
this may mean that it
was not as close to the identity as expected. 
If we simply let the algorithm proceed to iteration $k + 1$ with $\betain{k+1} =
\betaout{k}$, the next step is to compute $H_k(\betain{k+1})$.
Notice that $H_k(\betain{k+1})=H_k(\betaout{k})$ is
exactly the matrix we had assumed was identity.  Now we actually compute it and
its Cholesky factor and then use this new Cholesky factor to define a new basis
$\bp_{k + 1}$---which should be closer to orthonormal---and a new search
direction $\bd_{k + 1}$. 
For numerical results see
Section \ref{subsubsec:static-results}, in particular
\tabref{tab:near-boundary-chol-vs-GS}, below.

Other choices for changing the basis can be generated via the singular value
decomposition (SVD).  To wit, if $U \Lambda U^\T$ is the SVD of
$H_{k-1}(\betain{k})$ and $O$ is any $(N+1) \times (N+1)$ orthogonal matrix,
then $L_k=U \Lambda^{1/2} O^T$ satisfies  \eqref{eq:H=LL'}.  (For example, the
choice $O=U$ makes $L_k$ the symmetric square root of $H_{k-1}(\betain{k})$.) We
have found that the change-of-basis defined using Cholesky factorization
performs just as well as that defined by the SVD with $O=I$ in the sense that
the number of problems solvable without regularization is nearly the same.
Therefore, since the SVD gives a non-triangular basis, is more expensive, and is
also harder to parallelize, we conclude that the Cholesky factorization is a
better choice for our problem.

\subsection{Computing the stopping criterion}

At each iteration, the two conditions used in the stopping criterion in
\eqref{eq:stopping} are computed in the original Legendre basis $\bm$. We use
this basis because it is physically relevant for the kinetic equation and using
the original basis for flux calculations is simpler to implement in a parallel
setting because the adaptive basis varies with each spatial cell.

The calculation of the gradient in the Legendre basis for the stopping criterion
can
be as simple as $\bg(\bsalpha_k) = T_k \bg_{k}(\betaout{k})$, a calculation that
takes only $(N+1)^2$ multiplications. Alternatively, we can compute the gradient
by first switching back to the multipliers in the Legendre basis by computing
$\bsalpha_k = S_k^\T \betaout{k}$ and then
 \begin{equation}
  \bg(\bsalpha_k) = \Vint{\exp(\bsalpha_k^\T \bm)} - \bu .
 \end{equation}
This computation is significantly more expensive (requiring $(N + 1 + 2\nQ)(N +
1)$ multiplications, to leading order, and $\nQ$ exponential evaluations). 
While in exact arithmetic the results are identical, the latter has the
advantage of consistency: the same $\bsalpha_k$ is used in the computation of
the flux in \eqref{eq:semidisc}.

The estimated upper bound on $\gamma$ in the stopping criterion
(\ref{eq:stopping}), whose computation amounts to estimating the maximum value
of the polynomial $\bd_k(\betaout{k})^\T \bp_k = \bd(\bsalpha_k)^\T \bm$, can be
computed in either basis.  The Newton direction $\bd_k(\betaout{k})$ can be
converted back to basis $\bm$ using $T_k$ to use \eqref{eq:L1bnd} directly, or
we can modify \eqref{eq:L1bnd} to use the one-norm of $\bd_k(\betaout{k})$ and
the maximum of the basis polynomials in $\bp_k$ on the quadrature nodes,
$\max_{ij} |p_i(\mu_j)|$ (see Theorem \ref{thm:invariance}). We chose the former
though we did not notice a significant difference between the two options in the
performance of the optimizer.

\subsection{Returning to the Legendre basis}

Solving an optimization problem in a changing basis requires careful
bookkeeping.  In \algref{alg:change-of-basis} we choose to update both the
matrix $T_k$
defining the change of basis and the 
$(N+1)\times\nQ$ matrix $P_k$ of basis polynomial values at
quadrature points, even though it is only strictly necessary to update one of
them.%
\footnote{Indeed, at every iteration, each matrix can be obtained from the
other ($P_k = T_k^{-1} M$ (where $M$ holds the values
of the original basis polynomials $\bm$ at the quadrature nodes), and $T_k^{-1}
= P_k M^\T (M M^\T)^{-1}$), and each matrix can be incrementally updated using
$L_k$.
}

Firstly, we choose to update $P_k$ because it is used repeatedly at
each iteration in quadratures during the line search. We choose to update $T_k$
as well (at a cost of $N^3$ multiplications per iteration) because it makes
the computations for the stopping criterion simpler, and this extra cost had
negligible effects on the total computation time for our implementation.



\subsection{Regularization}
\label{sect:regularization}

The algorithm in \cite{AHT10} deems an optimization problem `too difficult'
when the adaptive quadrature routine requires more points to estimate the
objective function than a user-prescribed limit.  In such cases, the
regularization parameter $r$ in (\ref{eq:v(r)}) is increased.

In Algorithm \ref{alg:r}, we instead increase the regularization parameter when
the optimization has not converged after $k_0$ iterations. As a result, in our
implementation, regularization is used less frequently than in \cite{AHT10}.  As
in \cite{AHT10}, we assume that $r_{\max}$, the highest value of $r$ used by the
algorithm, is such that all problems can be solved.  Therefore, when $r$ reaches
$r_{\max}$, we continue the optimization past $k_0$ iterations. While we have
been able to construct moments for which $r_{\max} = 10^{-4}$ is not large
enough to produce a numerically solvable optimization problem (see Section
\ref{subsubsec:static-results} below), we have never found problems this hard in
any of the benchmark simulations, where that value of $r_{\max}$ was used.

\section{Numerical Results}
\label{sec:results}

In this section we report on a series of numerical experiments we performed to
assess the performance of \algref{alg:r}.  These include experiments
with
(i) static problems for a fixed set of moments,
(ii) computation of a manufactured solution, and
(iii) the simulation of two well-known benchmark problems in radiative
transport.

Unless otherwise noted, we use the following parameter values:
\begin{equation}
 \begin{array}{|rcll|}
 \hline
  \tau &=& 10^{-9}\,, \quad & \mbox{upper bound
    for}\:\|\bg(\bsalpha_k)\| \: \mbox{in the stopping criterion,} \\
    \hline
  \veps_\gamma &=& 0.01\,, \quad & \mbox{upper bound on}\;\gmax-1 
   \;\mbox{to maintain realizability,} \\
   \hline
  \theta & = & 2.0\,, \quad & \mbox{slope limiting parameter,}\\ 
  \hline
  \chi & = & 0.5\,, \quad & \mbox{line search step size decrease parameter,}\\ 
  \hline
  \xi & = & 10^{-3}\,, \quad & \mbox{line search sufficient decrease
   parameter,}\\
   \hline
  \rseq & = & \{ 0, 10^{-8}, 10^{-6}, 10^{-4} \} & \mbox{sequence of
   regularization parameters to try}\\
  \hline
  k_{\max}&=& 200\,, \quad & \mbox{maximum number of iterations,} \\
  \hline
  \veps &=& 2^{-52}\,, \quad & \mbox{parameter used in line-search
   termination.} \\
  \hline
 \end{array} \nonumber
\end{equation}
When simulating \eqref{eq:semidisc}, we set 
\begin{equation}
 \dt = \frac{0.95}{1 + \veps_\gamma} \frac{2}{\theta + 2} \dx\,,
\label{eq:dt}
\end{equation}
which, in view of~(\ref{eq:gamma<1+eps}), satisfies the time-step 
restriction \eqref{eq:cfl}.%
\footnote{For all cases considered here $\sig{t} = O(1)$.  Thus the effect of $\sig{t}$ on the CFL condition in \eqref{eq:cfl} is accounted for by the ``safety factor" $0.95$.}

The initial multipliers for the optimization algorithm at $t = 0$ are those
corresponding to the isotropic distribution, $\betain{0} = (\log(u_0/2), 0,
\ldots, 0)^\T$, and the initial basis for each problem is the Legendre basis
$\bm$, so that $T_{-1} = I$, $P_{-1} = M$. At later times we begin each
optimization with the final multipliers and basis from the spacial cell's
optimization problem at the previous time step. Consequently, $T_k$ is always
lower-triangular, allowing us to use the gradient formula in Algorithm
\ref{alg:change-of-basis}.

As discussed in \secref{subsec:quad}, we use Curtis-Clenshaw quadrature to
approximate all angular integrals.  Following \cite{AHT10}, we approximate each
half interval $\mu \in [-1, 0]$ and $\mu \in [1, 1]$ separately since (due to
upwinding) integrals at cell edges in the numerical scheme have different forms
for each half-interval.  Except for \secref{subsubsec:expadapt}, we use an equal
number of quadrature points (i.e. $\nQ/2$) on each half-interval.

\subsection{Static Results} \label{subsec:static-results}
 
We perform two experiments using static problems---that is, problems for
which the moments are chosen, as opposed to being generated by the solution of a
partial differential equation.

\subsubsection{Adaptive vs. Fixed Basis with Different Quadratures}
\label{subsubsec:expadapt}

In our first experiment, we use the following $\M_{15}$ moment vector that was
encountered in \cite{AHT10}: \begin{align}
 &\begin{tabular}{r@{.}lr@{.}lr@{.}lr@{.}l}
    $\bu \, = \, [$  1&0, &
       0&837872568, &
       0&572819692, &
       0&294071376, \\
       0&079519254, &
      -0&034894762, &
      -0&060428124, &
      -0&037077987, \\
      -0&006145576, &
       0&009337451, &
       0&007920869, &
       0&000075451, \\
      -0&004350212, &
      -0&002832808, &
       0&001074657, &
       0&003022835$]^\T$ 
  \end{tabular} \label{eq:u15}
\end{align}
to compare the fixed- and adaptive-basis methods. In \figref{fig:old-vs-new}, we
show the results from attempting to solve the dual problem with this value of
$\bu$ using a fixed-basis method (in the Legendre basis) and the adaptive-basis
method of \algref{alg:r}.  In each case, the initial multiplier vector
corresponds to the isotropic distribution and no regularization is used, i.e.
$\rseq = \{ 0 \}$. Each square pixel in \figref{fig:old-vs-new} displays the
number of iterations required in \algref{alg:r} for a particular choice of
quadrature. White pixels indicate that the optimization was unable to converge
within 200 iterations.

For the adaptive basis method, the algorithm converges for every tested
quadrature with at least $45$ nodes on $[-1, 0]$ in $13$--$64$ iterations. This
suggests that $45$ Curtis-Clenshaw nodes are needed to describe the structure in
the ansatz on $\mu \in [-1, 0]$ near the solution $\alphahat$. The fixed basis
method, on the other hand, is highly unpredictable and frequently does not
converge within 200 iterations.

\begin{figure}
 \centering
 \subfigure[Using a fixed Legendre basis.]
  {\label{subfig:old}
  \includegraphics[width = .4\textwidth]
   {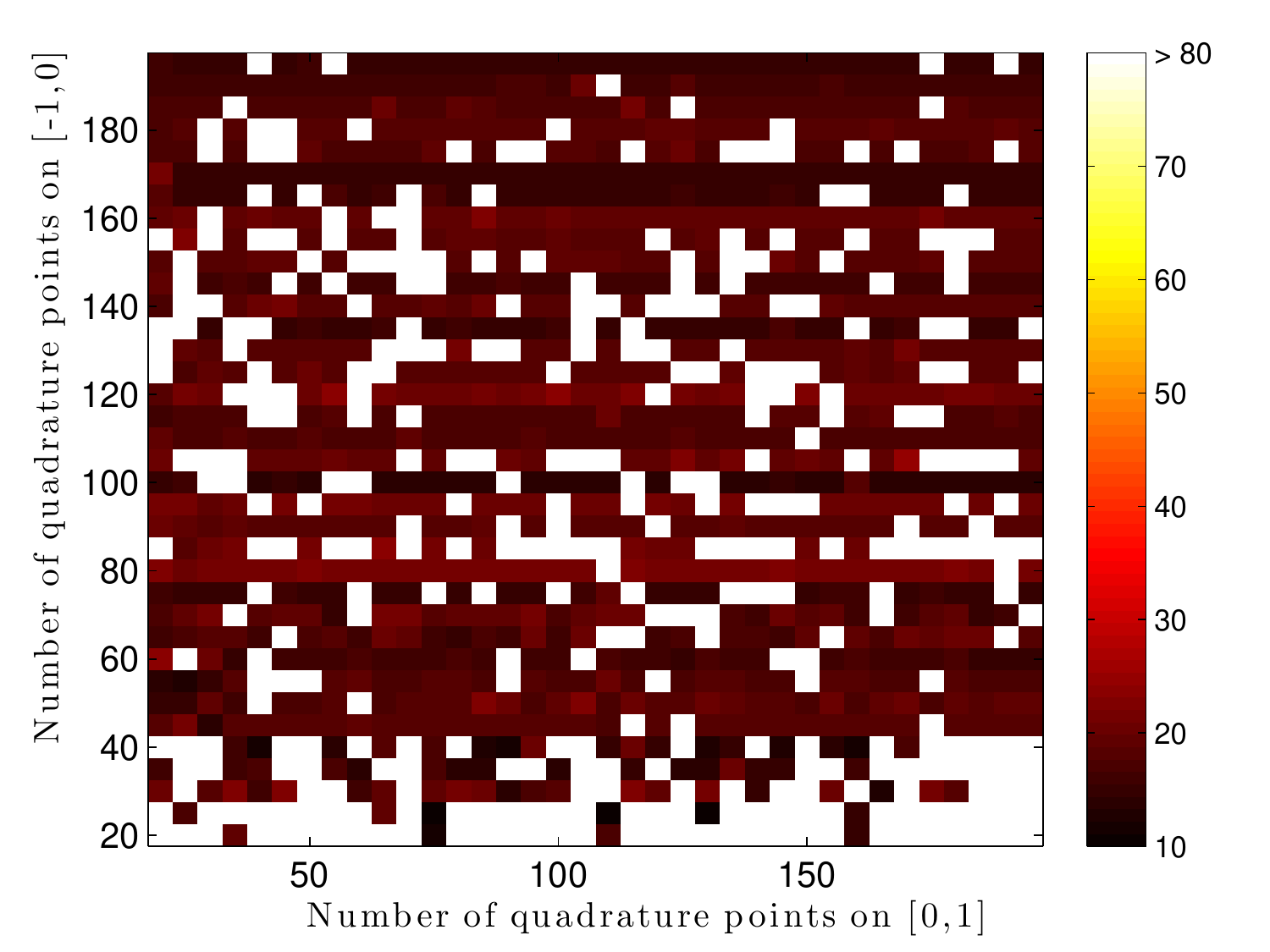}}
 \hspace{.2in}
 \subfigure[Using the adaptive-basis method.]{\label{subfig:new}
  \includegraphics[width = .4\textwidth]
   {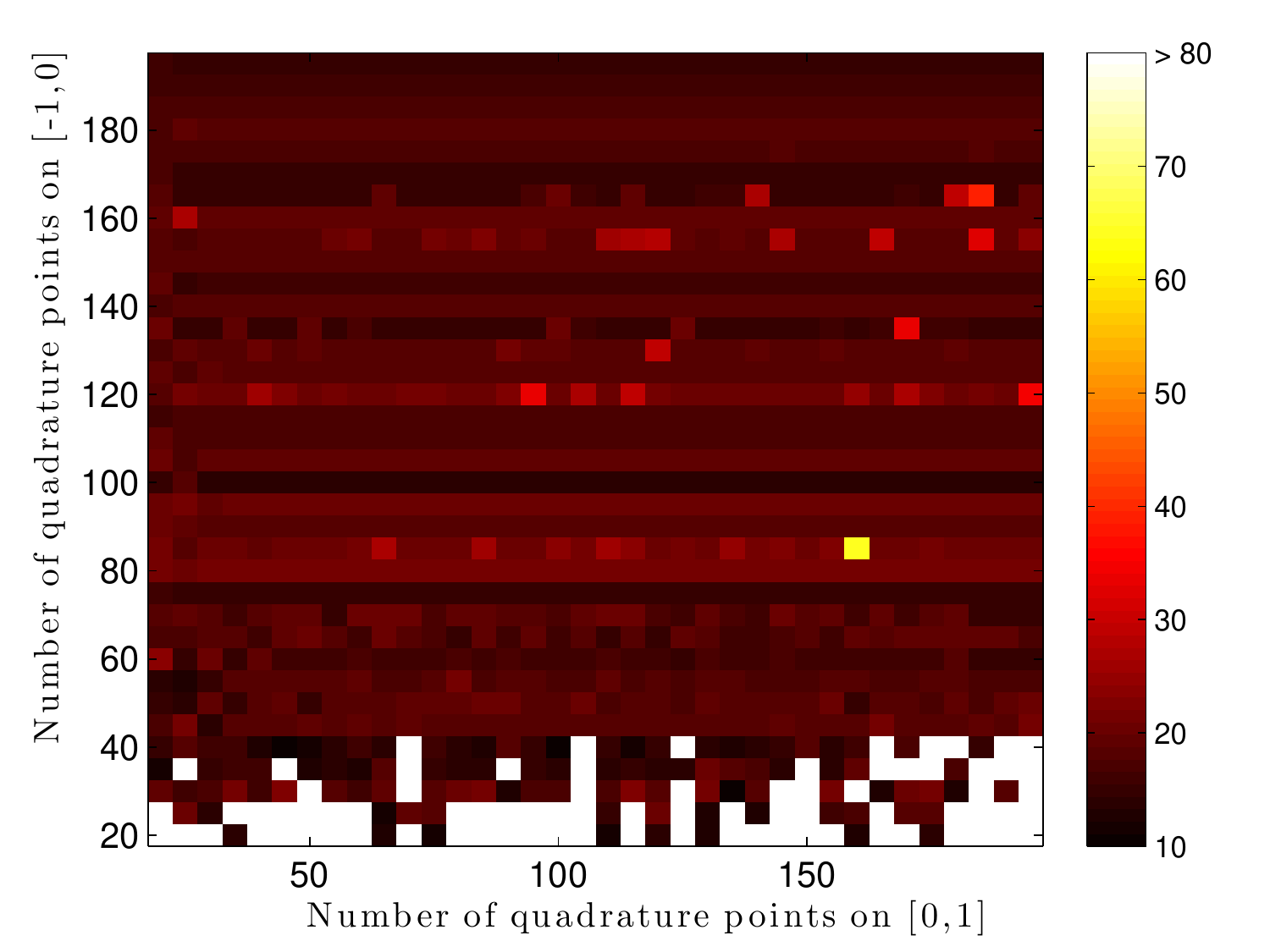}}
 \caption{Number of iterations used by the fixed-basis method 
and the adaptive-basis method on the moments given in \eqref{eq:u15} 
using Curtis-Clenshaw quadratures with a varying number of nodes.} 
 \label{fig:old-vs-new}
\end{figure}

\subsubsection{Approaching the Boundary of Realizability}
\label{subsubsec:static-results}

Next we show that, compared to a fixed basis, the adaptive-basis method allows
us to solve optimization problems closer to the realizable boundary $\p \Rm$.
For $N \ge 2m$, moments given by
\begin{equation}
 \bu = \Vint{\bm \sum_{i = 1}^m c_i \delta(\mu - \nu_i)}
  = \sum_{i = 1}^m c_i \bm(\nu_i)
\label{eq:boundary-moment}
\end{equation}
(where $c_i \ge 0$ and $\nu_i \in [-1, 1]$ ) lie on $\p \Rm$
\cite{Curto-Fialkow-1991}.
For $\ell \in \{0,1,2,\ldots \}$, let 
$r_{\ell} = 2^{-\ell}$.  
Then for any $\bu \in \p\Rm$, the sequence of regularized moments
(cf. \eqref{eq:v(r)})
\begin{equation}
 \bv(\bu,r_{\ell}) = (1 - r_{\ell}) \bu + r_{\ell} Q \bu \,,
  \quad \ell \in \{0,1,2,\ldots \}\:,
\label{eq:sequence-to-boundary}
\end{equation}
approaches $\p \Rm$ as $\ell \rightarrow \infty$. Let $\ell_{\rm A}$ be the
largest value of $\ell$ such that the adaptive-basis method converges with input
moment $\bv(\bu,r_{\ell})$ for all $\ell \le \ell_{\rm A}$, and let $\ell_{\rm
F}$ be defined in a similar way for the fixed-basis method. Then $r_{\ell_{\rm
A}}$ and $r_{\ell_{\rm F}}$ give an indication 
of how much regularization is needed to solve different problems near the
realizable boundary.

We performed experiments to compute $\ell_{\rm A}$ and $\ell_{\rm F}$ with moment
vectors of length $13$ ($N = 12$) that are generated via
\eqref{eq:boundary-moment} using a combination of $m = 6$ delta functions. 
Table~\ref{tab:u-boundary} gives the strengths $c_i$ and 
locations $\nu_i$ of the
delta functions generating example moments $\bu^{(1)}, \ldots, \bu^{(6)}$, each
of which lies on $\p \Rm$. The moments $\bu^{(1)}, \ldots, \bu^{(4)}$ are chosen
to show the effects of changing the strengths $c_i$ and the distance between
locations of the delta functions, while $\bu^{(5)}$ and $\bu^{(6)}$ are chosen
to illustrate the case when quadrature nodes are co-located with the delta
functions generating the moments on the boundary.
The delta functions generating $\bu^{(5)}$ are located at the 4-th, 7-th, 10-th,
13-th, 14-th, and 15-th nodes of the 20-point Curtis-Clenshaw quadrature
over the interval $[-1, 0]$;  for $\bu^{(6)}$, the locations $\nu_1$, $\nu_3$,
$\nu_4$, $\nu_5$, and $\nu_6$ correspond to the 4-th, 9-th, 12-th, 15-th, and
17-th nodes, respectively, of the same quadrature, while $\nu_2$ is the 54-th
node of the 153-point Curtis-Clenshaw quadrature over the interval $[-1, 0]$.

Values of $r_{\ell_{\rm A}}$ and $r_{\ell_{\rm F}}$ are displayed in Table
\ref{tab:near-boundary12}.  These values are computed using several different
quadratures sizes.  The values of $\nQ$ are chosen so that the quadratures are
nested.  (The nodes of the Curtis-Clenshaw quadrature of order $2n - 1$ include
all nodes of the quadrature of order $n$.)  The table shows that the
adaptive-basis method nearly always uses a smaller regularization parameter than
the fixed-basis method, in many cases by two orders of magnitude.
\begin{table}
\caption{Boundary moments used for tests in Table \ref{tab:near-boundary12}
below.}
\begin{center}
\begin{tabular}{cc|c|c|c|c|c|c}
\label{tab:u-boundary}
 & $i$ & 1 & 2 & 3 & 4 & 5 & 6 \\ \midrule\midrule
\multirow{2}{*}{$\bu^{(1)}$} & $\nu_i$  & 0.2 & 0.3 & 0.4 & 0.5 & 0.6 & 0.7 \\
 & $c_i$ & 0.167 & 0.167 & 0.167 & 0.167 & 0.167 & 0.167 \\ \midrule
\multirow{2}{*}{$\bu^{(2)}$} & $\nu_i$  & 0.2 & 0.3 & 0.4 & 0.5 & 0.6 & 0.7 \\
 & $c_i$ & 0.0833 & 0.0833 & 0.0833 & 0.333 & 0.0833 & 0.333 \\ \midrule
\multirow{2}{*}{$\bu^{(3)}$} & $\nu_i$  & 0.2 & 0.4 & 0.6 & 0.88 & 0.89 & 0.9 \\
 & $c_i$ & 0.0833 & 0.0833 & 0.0833 & 0.333 & 0.0833 & 0.333 \\ \midrule 
\multirow{2}{*}{$\bu^{(4)}$} & $\nu_i$  & -0.8 & -0.5 & -0.1 & 0.59999 & 0.6 &
0.8 \\
 & $c_i$ & 0.167 & 0.167 & 0.167 & 0.167 & 0.167 & 0.167 \\ \midrule
\multirow{2}{*}{$\bu^{(5)}$} & $\nu_i$  & -0.94 & -0.773 & -0.541 & -0.299 &
-0.227 & -0.161 \\
 & $c_i$ & 0.417 & 0.0417 & 0.0417 & 0.0417 & 0.417 & 0.0417 \\ \midrule 
\multirow{2}{*}{$\bu^{(6)}$} & $\nu_i$  & -0.94 & -0.729 & -0.623 & -0.377 &
-0.161 & -0.0603 \\
 & $c_i$ & 0.167 & 0.167 & 0.167 & 0.167 & 0.167 & 0.167
\end{tabular}
\end{center}
\end{table}
\begin{table}
\caption{Values of the smallest regularization parameters
giving convergence of \algref{alg:r} for the moments in Table
\ref{tab:u-boundary}: $r_{\ell_{\rm{A}}}$ (adaptive basis) and
$r_{\ell_{\rm{F}}}$ (fixed basis).}
\begin{center}
\begin{tabular}{ccc|cc|cc|cc|cc|cc}
\label{tab:near-boundary12}
 & \multicolumn{2}{c}{$\bu^{(1)}$} & \multicolumn{2}{c}{$\bu^{(2)}$} &
\multicolumn{2}{c}{$\bu^{(3)}$} & \multicolumn{2}{c}{$\bu^{(4)}$} &
\multicolumn{2}{c}{$\bu^{(5)}$} & \multicolumn{2}{c}{$\bu^{(6)}$} \\
\cmidrule(r){2-3} \cmidrule(r){4-5}   \cmidrule(r){6-7}
\cmidrule(r){8-9} \cmidrule(r){10-11}   \cmidrule(r){12-13}
$\nQ$ & $r_{\ell_{\rm{A}}}$ & $r_{\ell_{\rm{F}}}$ & $r_{\ell_{\rm{A}}}$ &
$r_{\ell_{\rm{F}}}$ & $r_{\ell_{\rm{A}}}$ & $r_{\ell_{\rm{F}}}$ &
$r_{\ell_{\rm{A}}}$ & $r_{\ell_{\rm{F}}}$ & $r_{\ell_{\rm{A}}}$ &
$r_{\ell_{\rm{F}}}$ & $r_{\ell_{\rm{A}}}$ & $r_{\ell_{\rm{F}}}$ \\
\midrule
 40 & 1.5e-5 & 1.2e-4 & 1.2e-4 & 1.2e-4 & 3.1e-2 & 3.1e-2 & 1.2e-1 &
1.2e-1 & 3.1e-2 & 1.5e-5 & 7.5e-9 & 6.0e-8 \\
 78 & 1.2e-7 & 1.2e-7 & 6.0e-8 & 1.5e-5 & 7.6e-6 & 1.5e-5 & 1.6e-2 &
1.6e-2 & 1.9e-9 & 3.8e-6 & 7.5e-9 & 7.8e-3 \\
154 & 1.2e-7 & 1.5e-5 & 3.0e-8 & 1.6e-2 & 2.4e-7 & 3.8e-6 & 3.9e-3 &
1.6e-2 & 3.0e-8 & 1.9e-6 & 7.5e-9 & 1.5e-5 \\
306 & 1.2e-7 & 3.1e-5 & 1.2e-7 & 7.6e-6 & 2.4e-7 & 1.9e-6 & 9.8e-4 &
3.9e-3 & 2.4e-7 & 1.5e-5 & 7.5e-9 & 1.9e-6 \\
610 & 1.2e-7 & 3.1e-5 & 7.6e-6 & 9.5e-7 & 4.8e-7 & 6.1e-5 & 2.4e-4 &
3.9e-3 & 9.5e-7 & 2.4e-4 & 7.5e-9 & 4.8e-7
\end{tabular}
\end{center}
\end{table}

We repeated the boundary-moment tests of Table \ref{tab:near-boundary12} using
the modified Gram-Schmidt method instead of Cholesky factorization, in order to
assess whether the added stability enables the solution of dual problems for
moments closer to the realizable boundary. The results in Table
\ref{tab:near-boundary-chol-vs-GS} show that the modified Gram-Schmidt algorithm
does not provide a significant advantage.

\begin{table}
\caption{Same as Table \ref{tab:near-boundary12} above, but here we compare the
use of Cholesky factorization (C) to that of modified Gram-Schmidt (GS) in the
adaptive-basis method.}

\begin{center}
\begin{tabular}{ccc|cc|cc|cc|cc|cc}
\label{tab:near-boundary-chol-vs-GS}
 & \multicolumn{2}{c}{$\bu^{(1)}$} & \multicolumn{2}{c}{$\bu^{(2)}$} &
\multicolumn{2}{c}{$\bu^{(3)}$} & \multicolumn{2}{c}{$\bu^{(4)}$} &
\multicolumn{2}{c}{$\bu^{(5)}$} & \multicolumn{2}{c}{$\bu^{(6)}$} \\

\cmidrule(r){2-3} \cmidrule(r){4-5}   \cmidrule(r){6-7}
\cmidrule(r){8-9} \cmidrule(r){10-11}   \cmidrule(r){12-13}

$\nQ$ & $r_{\ell_{\rm{C}}}$ & $r_{\ell_{\rm{GS}}}$ & $r_{\ell_{\rm{C}}}$ &
$r_{\ell_{\rm{GS}}}$ & $r_{\ell_{\rm{C}}}$ & $r_{\ell_{\rm{GS}}}$ &
$r_{\ell_{\rm{C}}}$ & $r_{\ell_{\rm{GS}}}$ & $r_{\ell_{\rm{C}}}$ &
$r_{\ell_{\rm{GS}}}$ & $r_{\ell_{\rm{C}}}$ & $r_{\ell_{\rm{GS}}}$ \\

\midrule

 40 & 1.5e-5 & 7.6e-6 & 1.2e-4 & 1.2e-4 & 3.1e-2 & 3.1e-2 & 1.2e-1 &
1.2e-1 & 3.1e-2 & 1.2e-10 & 7.5e-9 & 7.5e-9 \\
 78 & 1.2e-7 & 1.2e-7 & 6.0e-8 & 6.0e-8 & 7.6e-6 & 7.6e-6 & 1.6e-2 &
1.6e-2 & 1.9e-9 & 9.3e-10 & 7.5e-9 & 7.5e-9 \\
154 & 1.2e-7 & 1.2e-7 & 3.0e-8 & 6.0e-8 & 2.4e-7 & 4.8e-7 & 3.9e-3 &
3.9e-3 & 3.0e-8 & 1.2e-7  & 7.5e-9 & 7.5e-9 \\
306 & 1.2e-7 & 1.2e-7 & 1.2e-7 & 1.2e-7 & 2.4e-7 & 1.2e-7 & 9.8e-4 &
9.8e-4 & 2.4e-7 & 4.8e-7  & 7.5e-9 & 7.5e-9 \\
610 & 1.2e-7 & 1.2e-7 & 7.6e-6 & 1.2e-7 & 4.8e-7 & 4.8e-7 & 2.4e-4 &
2.4e-4 & 9.5e-7 & 4.8e-7  & 7.5e-9 & 7.5e-9
\end{tabular}
\end{center}
\end{table}

\subsection{Results on a manufactured-solution testbed}

In our third test, we assess the effect of regularization on the accuracy of the
solution to the moment system. In general, this is difficult to measure since
the true solution is generally unknown, and we are unable to compute
high-resolution approximations without regularization. As an alternative,
we use the method of manufactured solutions
\cite{Knupp-Salari-2000,Roache-2002}.  Following this approach, we solve
numerically the system
\begin{equation}
 \p_t \bu + \p_x \bff(\bu) = \p_t \bw + \p_x \bff(\bw) \,,
\label{eq:moment-sys-source}
\end{equation}
where $\bw$ is a specified target solution. For simplicity, we set
\begin{equation} \bw(x, t) := \Vint{\bm \exp(\bsalpha(x, t)^\T \bm)} \,, \quad
 (x, t) \in [-1, 1] \times [0, \tf] \:,
 \label{eq:v_man}
\end{equation}
where%
\begin{subequations}
\begin{align}
 \alpha_1(x, t) &= 0.1 + \frac{K}{2}\left( \cos(\pi (x - t)) + 1 \right) 
 \:, \qquad  \qquad \quad K>0 \:,
 \label{eq:alpha1}
  \\
 \alpha_0(x, t) &= \log\left( \frac{(1 + \frac{1}{2}\cos(\pi (x - t)))
  \alpha_1(x, t)}{2 \sinh(\alpha_1(x, t))} \right) \:. 
  \label{eq:manufactured-alpha0}
 \end{align}
\end{subequations}
and $\alpha_2(x, t) \equiv \alpha_3(x, t) \equiv  \ldots \equiv
\alpha_N(x,t)\equiv 0$, so that integrals of the form $\Vint{\mu^k
\exp(\bsalpha^\T \bm)}$ can be computed explicitly. In particular, 
\begin{equation}
w_0(x, t) = 1 + \frac{1}{2}\cos(\pi (x - t)).
\end{equation}

It is clear from \eqref{eq:alpha1} that $\alpha_1$ is always positive, which
means that particles are always moving to the right.  Meanwhile, the parameter
$K$ controls the distance between $\bw$ and $\p \Rm$.  Indeed, as $K$ increases,
the ansatz $G_{\alphahat(\bw)}$ looks more and more like a single delta function
at $\mu = 1$, particularly when $x = t$, where $\alpha_1(x, t)$ reaches its
maximum.  The offset value $0.1$ is included in \eqref{eq:alpha1} order to bound
$\alpha_1$ away from zero, where the exact evaluation of integrals of the form
$\Vint{\mu^k \exp(\alpha_1 \mu)}$ is numerically unstable.  The profile of the
particle density $u_0$ is plotted at a few different times in Figure
\ref{fig:manufactured-u}.  For the errors and statistics presented below, we
use $N = 3$, $K = 10$, and $\tf = 0.2$.

The kinetic scheme requires cell averages of the known right-hand side of
\eqref{eq:moment-sys-source}. This requires the integral of $\p_t \bw$ over
$I_j$ as well as pointwise evaluation of the flux $ \bff(\bw(x, t)) = \Vint{\mu
\bm \exp(\bsalpha(x, t)^\T \bm)} $ at cell edges. For the target solution in
\eqref{eq:v_man}, the integral of $\p_t \bw$ does not have an analytical form,
so we approximate it using 16-point Curtis-Clenshaw quadratures on
$(x_{j-1/2},\, x_{j+1/2})$. Periodic boundary conditions are enforced using
ghost cells.


\begin{figure}
 \centering
 \includegraphics[width = .48\textwidth]{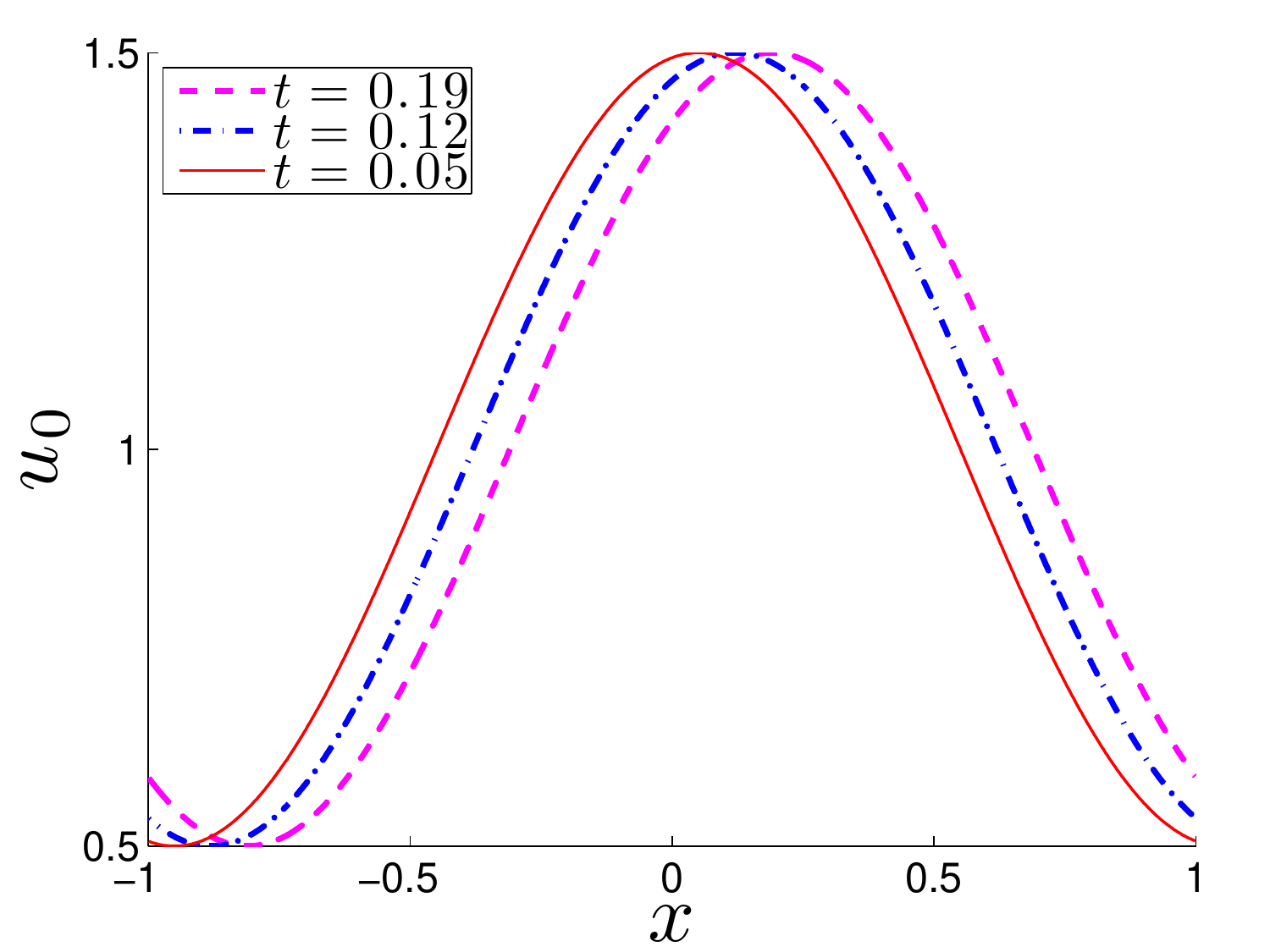}
 \caption{The particle density $u_0(x, t)$ for
  the manufactured-solution system
  \eqref{eq:moment-sys-source}-\eqref{eq:manufactured-alpha0} with $\M_3$ and $K
  = 10$, computed using $N_x = 800$ cells.}
 \label{fig:manufactured-u}
\end{figure}

Having the exact solution $\bw(x, t)$ allows us to calculate errors
for every numerical simulation.  We first interpolate the cell averages using
second-order affine reconstructions in each spatial cell:
\begin{equation}
 \bu_{\dx}(x) = \bu_j +(x-x_j)\frac{\bu_{j+1}-\bu_{j-1}}{2\dx}\,,
  \qquad x\in I_j, \qquad j \in \{1, \ldots, N_x\},
\label{eq:linear-reconstruction}
\end{equation}
where the moments are all taken at 
time $\tf$.
Then the $L^1$ and $L^\infty$ errors are given as
\begin{equation}
 \be^1_{\dx} 
:= \int_{-1}^1 | \bw(x, \tf)-\bu_{\dx}(x) | dx
\quand
   \be^\infty_{\dx} 
:= \max_{x\in[-1, 1]} |\bw(x, \tf)-\bu_{\dx}(x)| \:,
\label{eq:e1eInfty}
\end{equation}
respectively, where the absolute value is taken component-wise.  To approximate
the integral in $\be^1_{\dx}$, we split each spatial cell $I_j$ into 100 equally
sized subintervals, and then apply a twenty-point Gaussian quadrature on each
subinterval.  We approximate $\be^\infty_{\dx}$ with the maximum value of $|
\bw(x, \tf)-\bu_{\dx}(x) |$ over those same quadrature points.  Below, we only report
errors in the particle density $u_0$, i.e. the zero-th component of
$\be^1_{\dx}$ and $\be^\infty_{\dx}$. For $q \in \{1,\infty\}$, the order of
convergence $\bsnu$ between two successive grids of size $\dx_1$ and $\dx_2$ is
defined by the equality
\begin{equation}
 \be^q_{\dx_2} / \be^q_{\dx_1} 
  = \left({\dx_2} / {\dx_1}\right) ^{\bsnu} \:,
\end{equation}
where all operations are performed component-wise.

We now use the manufactured solution to compare the adaptive- and
fixed-basis methods. The results of Section \ref{subsec:static-results} suggest
that less regularization is needed with the adaptive-basis method.  Because
regularization introduces errors, we expect
the adaptive-basis solution to be more accurate.  To test this hypothesis, we
consider 
three cases:  the adaptive-basis method with $\nQ = 40$ quadrature
points (AB-Q40), the fixed-basis method with $\nQ = 40$ quadrature
points; (FB-Q40); and the fixed-basis method
with $\nQ = 240$ quadrature points (FB-Q240).  In all three cases, we use
Curtis-Clenshaw quadrature and
a regularization parameter $k_0 = 40$.

We first examine the difference in regularizations needed by each method.
\tabref{tab:AFstats-r-aggressive} shows (i) the fraction of problems regularized,
(ii) the average regularization parameter, and (iii) the statistic 
\begin{equation}
\label{eq:E}
E := \sum_{n=0}^{N_t - 1}  \sum_{m=1}^{2} \sum_{j=1}^{N_{x}}  r^{n,m}_{j} \, u_{j, 0}^{n,m} \,,
\end{equation} 
where $r^{n,m}_{j}$ is the value of $r$ used by the optimization algorithm
in cell $j$ at time $t^n = n\dt$ and Runge-Kutta stage $m$, and $u_{j,
0}^{n,m}$ 
is the corresponding cell average of the particle density $u_0$.
Within a single cell, the $L^1$-norm of the error introduced by
regularization is bounded by $r N u_0$, since (see~\eqref{eq:v(r)})

\begin{equation}
 \| \bu - \bv(\bu, r) \|_1 = r \| Q\bu - \bu \|_1
  = r \| (0, -u_1, \ldots , -u_N )^\T \|_1
  \le r N u_0 \:,
\end{equation}
where we have used the fact that,
since $\| m_i \|_{\infty} = 1$,
$|u_i| \le u_0$ for $i \in \{1, \ldots , N \}$. 
Thus $E$ gives us an estimate of the error introduced by
regularization. By each of these three metrics, the adaptive-basis method
uses less regularization than either fixed-basis method.

\begin{table}
\caption{Manufactured solution: Use of regularization for adaptive-basis with
$\nQ = 40$ (AB-Q40), fixed-basis with $\nQ = 40$ (FB-Q40), and fixed-basis with
$\nQ = 240$ (FB-Q240).  Recall that $E$ is defined in \eqref{eq:E}.}

\begin{center}
\begin{tabular}{cccccccccc}
\label{tab:AFstats-r-aggressive}

 & \multicolumn{3}{c}{fraction regularized} &
\multicolumn{3}{c}{mean $r$}
 & \multicolumn{3}{c}{$E$} \\
\cmidrule(r){2-4} \cmidrule(r){5-7} \cmidrule(r){8-10}
$N_x$ & AB-Q40 & FB-Q40 & FB-Q240 & AB-Q40 & FB-Q40 & FB-Q240
& AB-Q40 & FB-Q40 & FB-Q240 \\ \midrule


 100 & 4.55e-04 & 6.14e-03 & 4.32e-03 & 4.35e-12 & 4.89e-10 & 9.02e-10 &
8.89e-08 & 9.78e-06 & 1.82e-05 \\ 
 200 & 7.67e-03 & 2.99e-02 & 2.97e-02 & 5.25e-10 & 5.02e-09 & 5.47e-09 &
4.09e-05 & 3.85e-04 & 4.18e-04 \\ 
 400 & 6.98e-04 & 6.37e-03 & 6.90e-03 & 3.53e-11 & 5.32e-10 & 2.03e-09 &
1.06e-05 & 1.36e-04 & 4.76e-04 \\ 
 800 & 4.24e-04 & 2.72e-03 & 2.63e-03 & 5.10e-11 & 5.74e-10 & 5.87e-10 &
6.30e-05 & 7.09e-04 & 7.25e-04 \\ 
1200 & 6.77e-04 & 4.09e-03 & 3.96e-03 & 5.65e-11 & 8.63e-10 & 9.05e-10 &
1.57e-04 & 2.39e-03 & 2.51e-03 \\ 
1600 & 8.60e-04 & 5.27e-03 & 5.20e-03 & 9.63e-11 & 1.16e-09 & 1.13e-09 &
4.74e-04 & 5.72e-03 & 5.53e-03 \\ 
2000 & 1.08e-03 & 6.62e-03 & 6.36e-03 & 1.12e-10 & 1.50e-09 & 1.38e-09 &
8.57e-04 & 1.15e-02 & 1.06e-02

\end{tabular}
\end{center}
\end{table}

In \tabref{tab:AFerrors-aggressive}, we show $L^1$ and $L^\infty$ errors for
each method.  Along with the results in Table~\ref{tab:AFstats-r-aggressive}, these errors
suggest the following conclusions:
\begin{enumerate}
\itembf{The regularization causes the convergence to slow
significantly.}  While a different regularization strategy might improve these
results, we expect that both adaptive-basis and fixed-basis methods will
eventually stall.  This is because it is not possible to scale the
regularization with the mesh.  Indeed in many problems, the optimization may
become more difficult as the mesh is refined.  

\itembf{The adaptive-basis method shows better convergence properties.}  Even when the
convergence is not second-order, the adaptive-basis method has smaller error and
higher order of convergence.   This is to be expected since the adaptive-basis
method regularizes less and the values of the regularization  parameter
are smaller. Indeed, the
convergence of the fixed-basis method stalls around $1200$ cells, while the
adaptive-basis continues to converge (albeit slowly) up to $2000$ cells (the
finest mesh). This translates to an error that is a factor of eight smaller at
$2000$ cells.

\itembf{More quadrature points do not improve the fixed-basis method.} 
Increasing the number of quadrature points in the fixed-basis method from $40$
to $240$ does nothing to improve the convergence.  In addition, we have tried
using the adaptive-quadrature routine from \cite{AHT10}, which refines
the quadrature used to evaluate the dual objective function 
until the tolerance on a specified accuracy criterion is met. 
However, even with a tolerance of $10^{-13}$, the
adaptive-quadrature routine never refines the 40-point quadrature.  Thus the numerical
results are
the same as for FB-Q40, but the cost is much higher.
\end{enumerate}

\begin{table}
\caption{Manufactured solution: Convergence in space of $L^1$ and $L^\infty$
errors for adaptive-basis and fixed-basis optimization
methods. Errors are only computed for $u_0$.}

\begin{center}
\begin{tabular}{ccccccc}
\label{tab:AFerrors-aggressive}

 & \multicolumn{2}{c}{AB-Q40} & \multicolumn{2}{c}{FB-Q40} &
   \multicolumn{2}{c}{FB-Q240}\\
\cmidrule(r){2-3} \cmidrule(r){4-5} \cmidrule(r){6-7}
$N_x$ & $L^1$      & $\nu$ & $L^1$      & $\nu$ & $L^1$      & $\nu$ \\ \midrule


 100 & 5.63e-04 & ---  & 5.63e-04 &  ---  & 5.60e-04 &  ---  \\ 
 200 & 1.33e-04 & 2.08 & 1.33e-04 &  2.08 & 1.32e-04 &  2.09 \\ 
 400 & 3.11e-05 & 2.09 & 3.14e-05 &  2.08 & 3.26e-05 &  2.02 \\ 
 800 & 7.27e-06 & 2.10 & 9.02e-06 &  1.80 & 9.12e-06 &  1.84 \\ 
1200 & 3.58e-06 & 1.74 & 9.75e-06 & -0.19 & 9.92e-06 & -0.21 \\ 
1600 & 3.11e-06 & 0.49 & 1.40e-05 & -1.27 & 1.32e-05 & -0.99 \\ 
2000 & 2.83e-06 & 0.43 & 1.79e-05 & -1.09 & 1.73e-05 & -1.22 \\ \midrule

 & & & & & & \\

 & $L^\infty$ & $\nu$ & $L^\infty$ & $\nu$ & $L^\infty$ & $\nu$ \\ \midrule
 
 100 & 1.84e-03 & ---  & 1.84e-03 &  ---  & 1.85e-03 &  ---  \\ 
 200 & 7.75e-04 & 1.25 & 7.66e-04 &  1.26 & 7.47e-04 &  1.30 \\ 
 400 & 2.72e-04 & 1.51 & 2.79e-04 &  1.46 & 2.77e-04 &  1.43 \\ 
 800 & 8.35e-05 & 1.70 & 1.04e-04 &  1.42 & 1.03e-04 &  1.43 \\ 
1200 & 4.86e-05 & 1.33 & 1.36e-04 & -0.66 & 1.14e-04 & -0.26 \\ 
1600 & 3.73e-05 & 0.92 & 1.67e-04 & -0.72 & 1.66e-04 & -1.31 \\ 
2000 & 3.38e-05 & 0.45 & 2.31e-04 & -1.44 & 2.10e-04 & -1.05 

\end{tabular}
\end{center}
\end{table}

Finally, \tabref{tab:AF-cpu-iters} shows CPU times and mean iteration counts.
While the iteration counts are roughly the same for all, the CPU times for
AB-Q40 are $25\%-30\%$ larger than the times for FB-Q40. We attribute this
difference to the more expensive stopping criterion and extra matrix
computations (for example, updating the basis polynomials) in the adaptive-basis
method. 
The CPU times for FB-Q240 are roughly $10\%$ less than the times for AB-Q40.  We
expect that adding even more quadrature will eventually produce a more expensive
method, even though the convergence behavior will not improve.

\begin{table}
\caption{Manufactured solution time and iterations for adaptive-basis
with $\nQ = 40$ (AB-40), fixed-basis with $\nQ = 40$ (FB-40), and
fixed-basis with $\nQ = 240$ (FB-Q240).}

\begin{center}
\begin{tabular}{ccccccc}
\label{tab:AF-cpu-iters}

 & \multicolumn{3}{c}{CPU time (s)} & \multicolumn{3}{c}{mean iterations} \\
\cmidrule(r){2-4} \cmidrule(r){5-7}
$N_x$ & AB-Q40 & FB-Q40 & FB-Q240 & AB-Q40 & FB-Q40 & FB-Q240 \\ \midrule


 100 & 7.72e+00 & 6.28e+00 & 7.17e+00 & 2.45 & 2.46 & 2.46 \\ 
 200 & 4.31e+01 & 2.73e+01 & 3.13e+01 & 2.73 & 2.57 & 2.58 \\ 
 400 & 1.07e+02 & 8.60e+01 & 1.02e+02 & 2.04 & 2.05 & 2.06 \\ 
 800 & 3.58e+02 & 2.83e+02 & 3.24e+02 & 1.52 & 1.52 & 1.52 \\ 
1200 & 8.20e+02 & 6.42e+02 & 7.39e+02 & 1.53 & 1.53 & 1.53 \\ 
1600 & 1.48e+03 & 1.15e+03 & 1.31e+03 & 1.54 & 1.53 & 1.53 \\ 
2000 & 2.35e+03 & 1.80e+03 & 2.05e+03 & 1.55 & 1.54 & 1.54

\end{tabular}
\end{center}
\end{table}

From the results in this section we conclude that the adaptive-basis,
fixed-quadrature method, while more expensive than a fixed-basis,
fixed-quadrature method, reduces the need for regularization, thereby allowing
further convergence before regularization errors dominate.  We also see, from
the results using a fixed-basis, fixed-quadrature method with a much
higher-order quadrature, that simply adding quadrature points is not as
effective at reducing the need for regularization.

\subsection{Results on standard test problems}
We now revisit the two standard test problems considered in \cite{AHT10,
Hauck-2011}.
Throughout, we use $\nQ = 40$. We consider two values of $k_0$, the number of iterations of \algref{alg:r} before
regularization is increased. We use $k_0 = 6$ to illustrate a more aggressive
regularization scheme, and $k_0 = 40$ to illustrate a less aggressive
regularization scheme.  One would expect a more aggressive regularization scheme
to introduce more errors due to regularization, but also to solve the problem
more quickly since fewer iterations may be used.

\subsubsection{Plane-source problem}
\label{subsubsec:plane}

In this problem, we model particles in an infinite domain with a purely
scattering medium $\sig{t}=\sig{s}=1$.  We consider an initial condition
\begin{equation}
 \bu(x,0) = \delta (x) + 2F_{\rm floor}\:,
\end{equation}
where $F_{\rm floor}=0.5 \times 10^{-8}$ is used to keep moments away from the
realizable boundary. Although the problem is posed on an infinite domain, a
finite domain is required for practical computation and boundary conditions must
be specified. As in \cite{Hauck-2011,AHT10}, we approximate the infinite domain by the
interval $[\xL, \xR] = [-D/2,\, D/2]$, where $D := 2t_{\rm f} + 0.2$ is
chosen to ensure that the boundary has negligible effects on the solution.  At the
right and left ends of the boundary, we enforce the boundary conditions
\begin{equation}
 \bu(\xL,t) = \bu(\xR,t) = 2F_{\rm floor}
\end{equation}
for $t \ge 0$.

In \figref{fig:planeN15} we present results from a $N_x = 1000$ cell simulation
of the $\M_{15}$ system with the adaptive-basis method and $k_0 = 6$. The
corresponding figures for simulations using $k_0 = 40$ or the fixed-basis method
are qualitatively similar.  Sample profiles of the resulting particle density
$u_0(x,t)$ presented in \figref{subfig:planeN15u} agree with what was presented
in \cite{Hauck-2011,AHT10}. Figure \ref{subfig:planeN15iter} shows the iteration
profile.  The mean number of iterations, excluding trivially solvable isotropic
problems, was about $2.03$ for $k_0 = 6$, while for $k_0 = 40$ that mean was
$2.02$.  The iteration histogram in Figure \ref{subfig:planeN15iter-hist} shows
that indeed nearly $99\%$ of the optimization problems were solved in three
iterations or fewer with $k_0 = 6$. Figure \ref{subfig:planeN15r} shows the
points in space-time where moments were regularized.  These moments were
encountered as particles from the boundary push out into the vacuum along the
front $x = \pm t$ for $t > 0$.  

Figure \ref{fig:planeN15r-hist} shows a histogram of the
regularization parameter $r$.  For the adaptive basis, only $0.36\%$ of the
nontrivial optimization problems%
\footnote{%
Isotropic moments are not included in the statistics since the optimization is
trivial in this case.%
}
were regularized when $k_0 = 6$ and only $0.066\%$ when $k_0 = 40$. The figure
also includes the statistics from simulations using the fixed-basis method: 
when $k_0 = 6$, the results were similar, but when $k_0 = 40$, more
regularization was needed than when using the adaptive basis. All results are a
significant improvement over \cite{AHT10}, where the algorithm regularized about
$2.25\%$ of the problems.  This is due, in large part, to the change in
regularization strategy:  In \cite{AHT10}, regularization is applied when an
adaptive quadrature routine cannot satisfy a prescribed tolerance. Here we
instead regularize after a prescribed number of iterations, and this turns
out to be a less aggressive strategy.

\begin{figure}
 \centering
 \subfigure[Snapshots of the particle density $u_0(x,t)$ at $t=0.5$, $1.5$, and
  $3.5$.] {\label{subfig:planeN15u}
  \includegraphics[width = .4\textwidth]
  {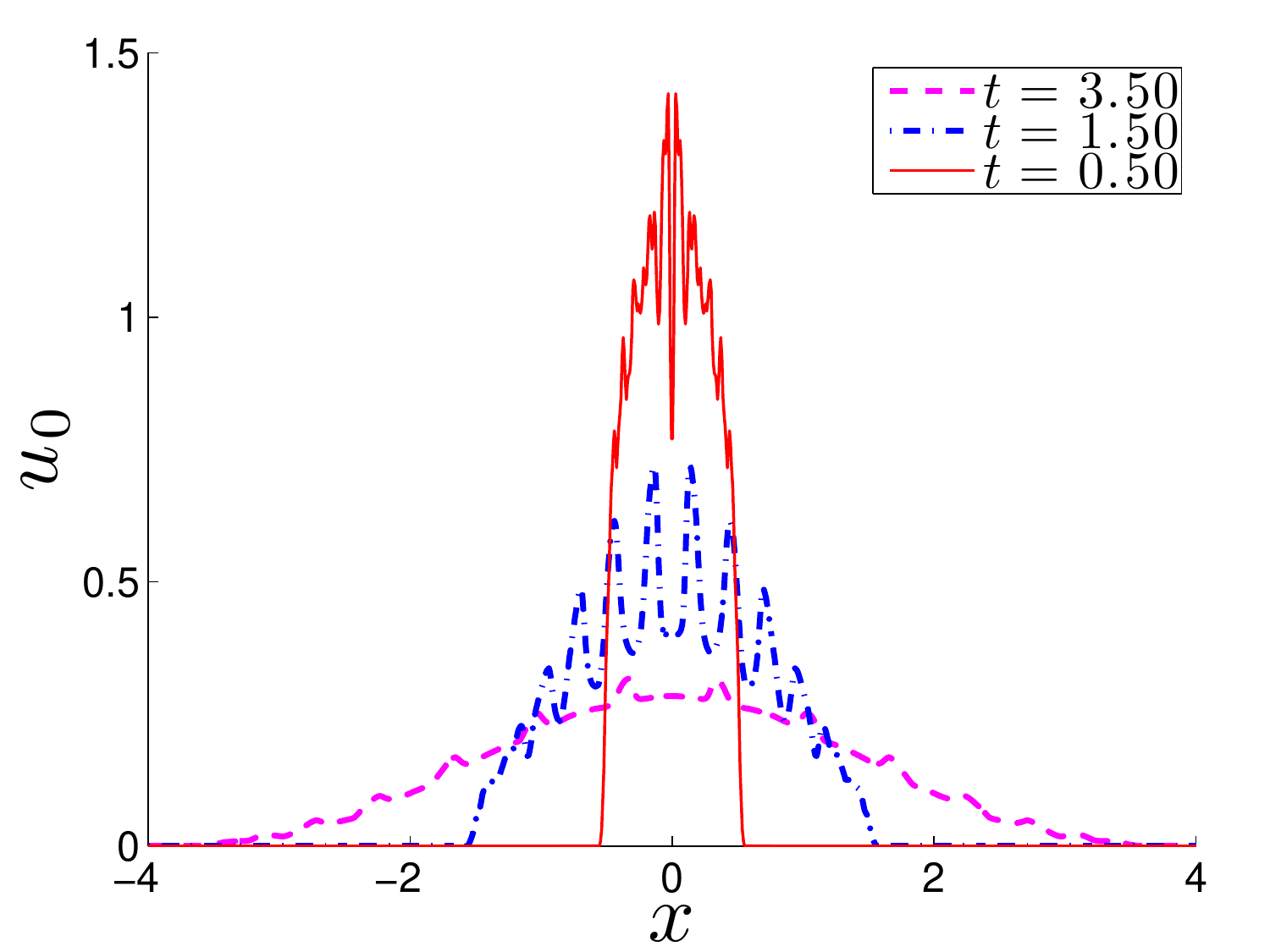}}
 \hspace{.2in}
 \subfigure[The total number of iterations (over the two Runge-Kutta stages) 
  needed to solve the optimization problem at each point in space and time. 
  The maximum number of iterations needed for one time step was 454 (off
  scale).]{\label{subfig:planeN15iter}
  \includegraphics[width = .4\textwidth]
   {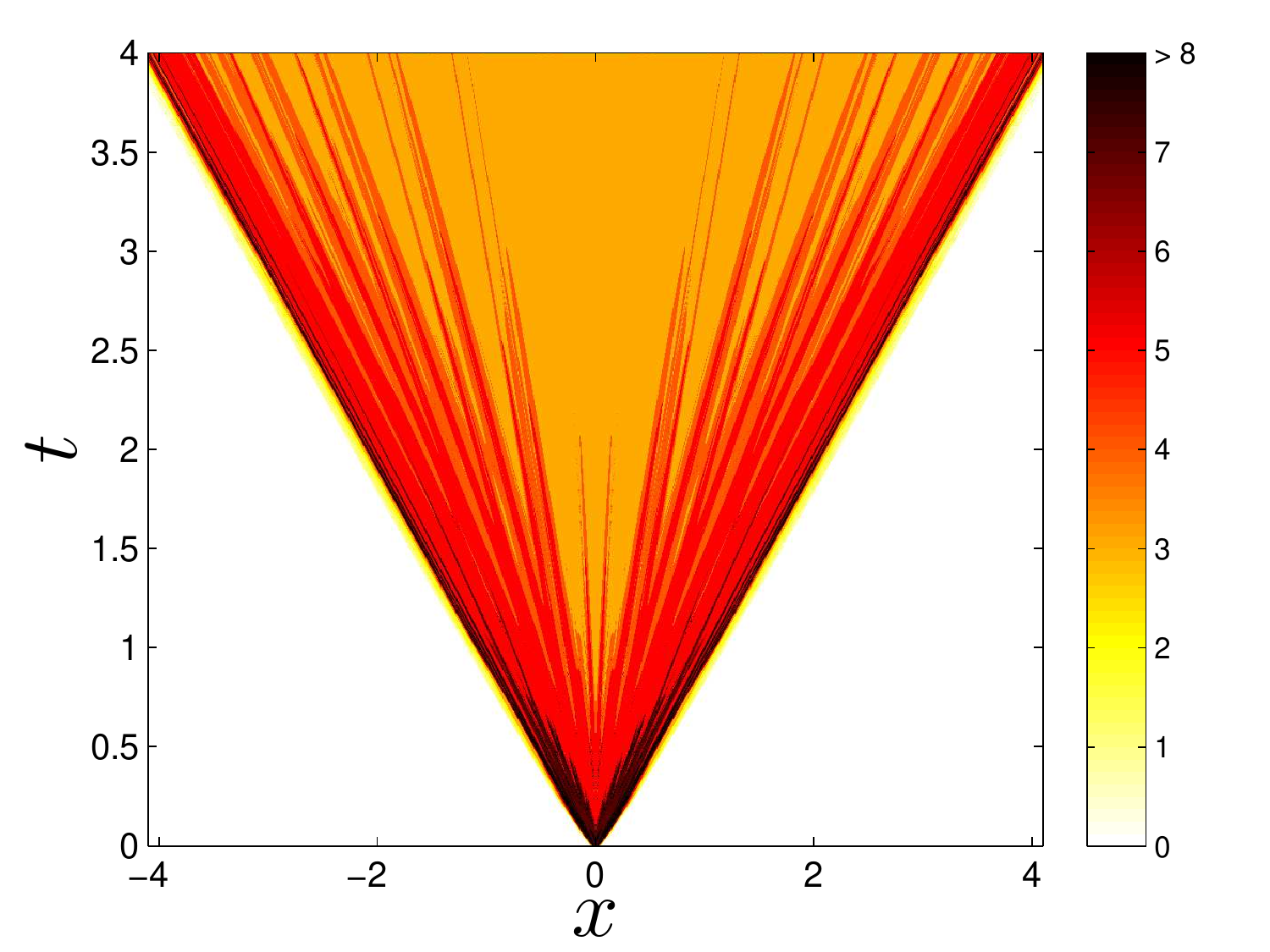}}\\
 \vspace{0.5in}
 \subfigure[A histogram of number of iterations needed to solve each
  optimization problem.  About $0.36\%$ of the problems needed more  six
  iterations.] {\label{subfig:planeN15iter-hist}
  \includegraphics[width = .4\textwidth]
   {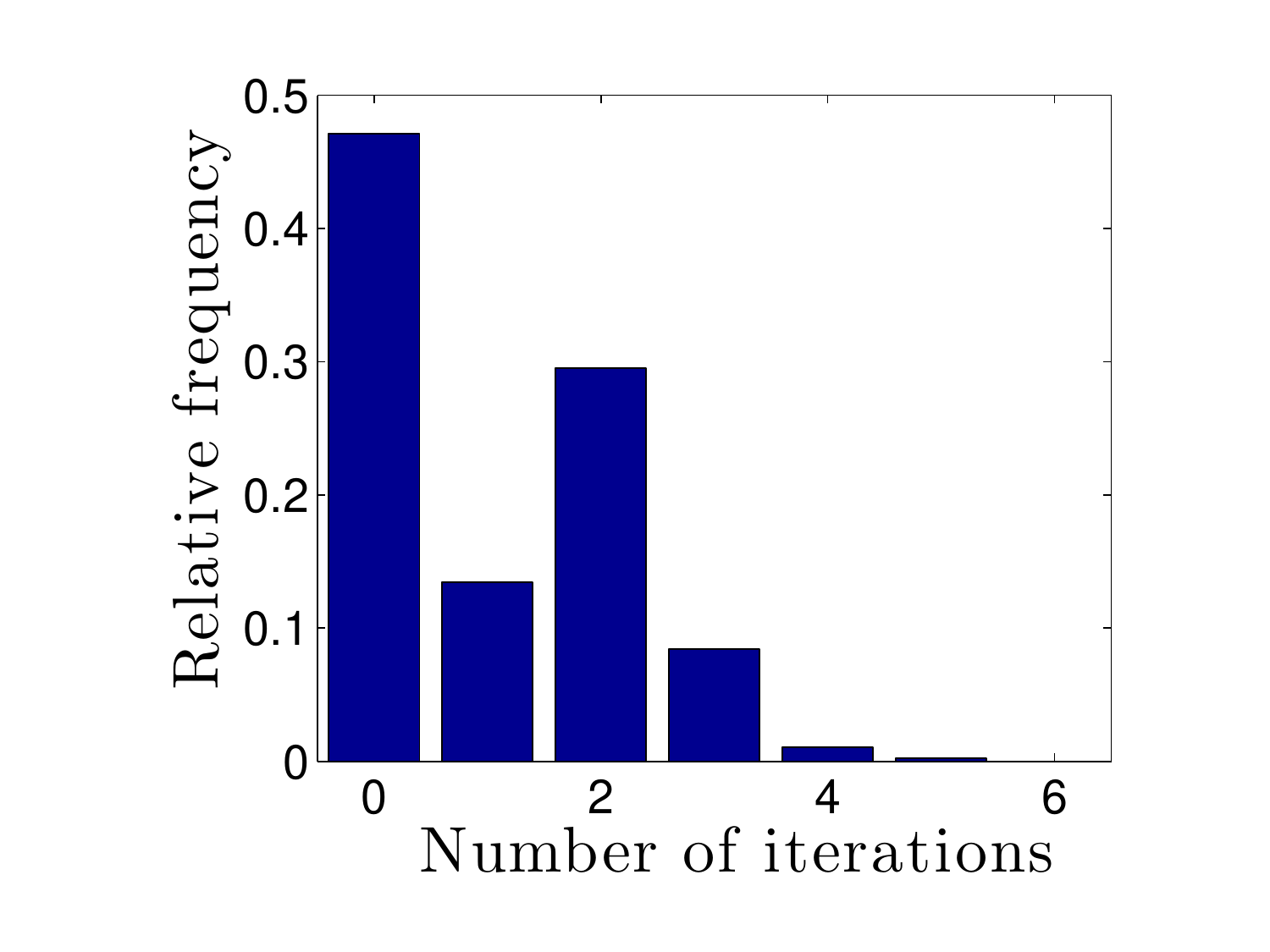}}
 \hspace{.2in}
 \subfigure[The locations of regularizations.]{\label{subfig:planeN15r}
  \includegraphics[width = .4\textwidth]
   {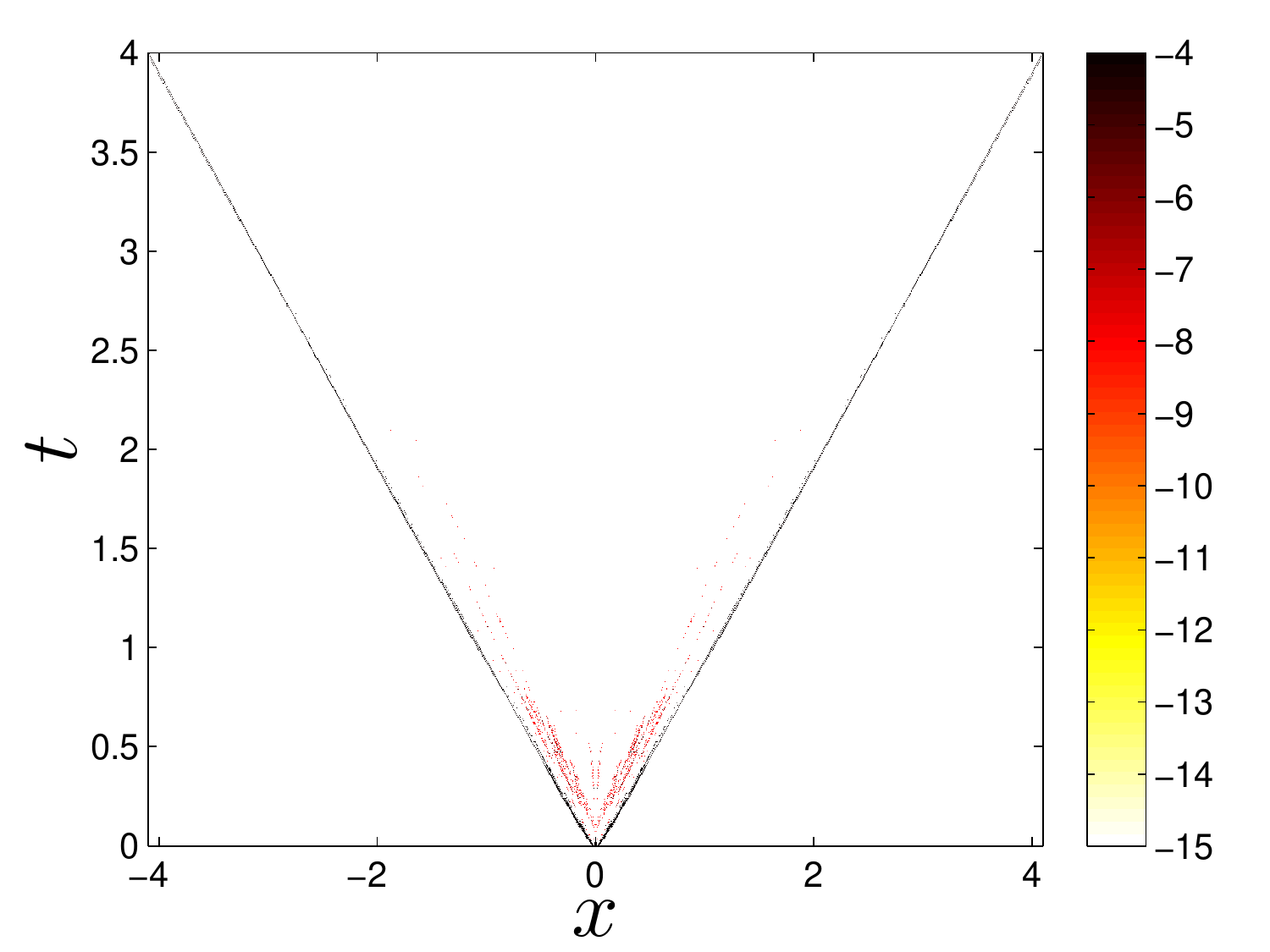}}
 \caption{Results for the $\M_{15}$ model of the plane-source problem 
  using the adaptive-basis method with
  $k_0 = 6$ and $N_x = 1000$ spatial cells.  For this
  simulation, excluding trivial cases such as cells with isotropic
  distributions, the optimization problem is solved about $1.1 \times 10^6$
  times.}
 \label{fig:planeN15}
\end{figure}

\begin{figure}
 \centering
 \includegraphics[width = 0.6\textwidth]{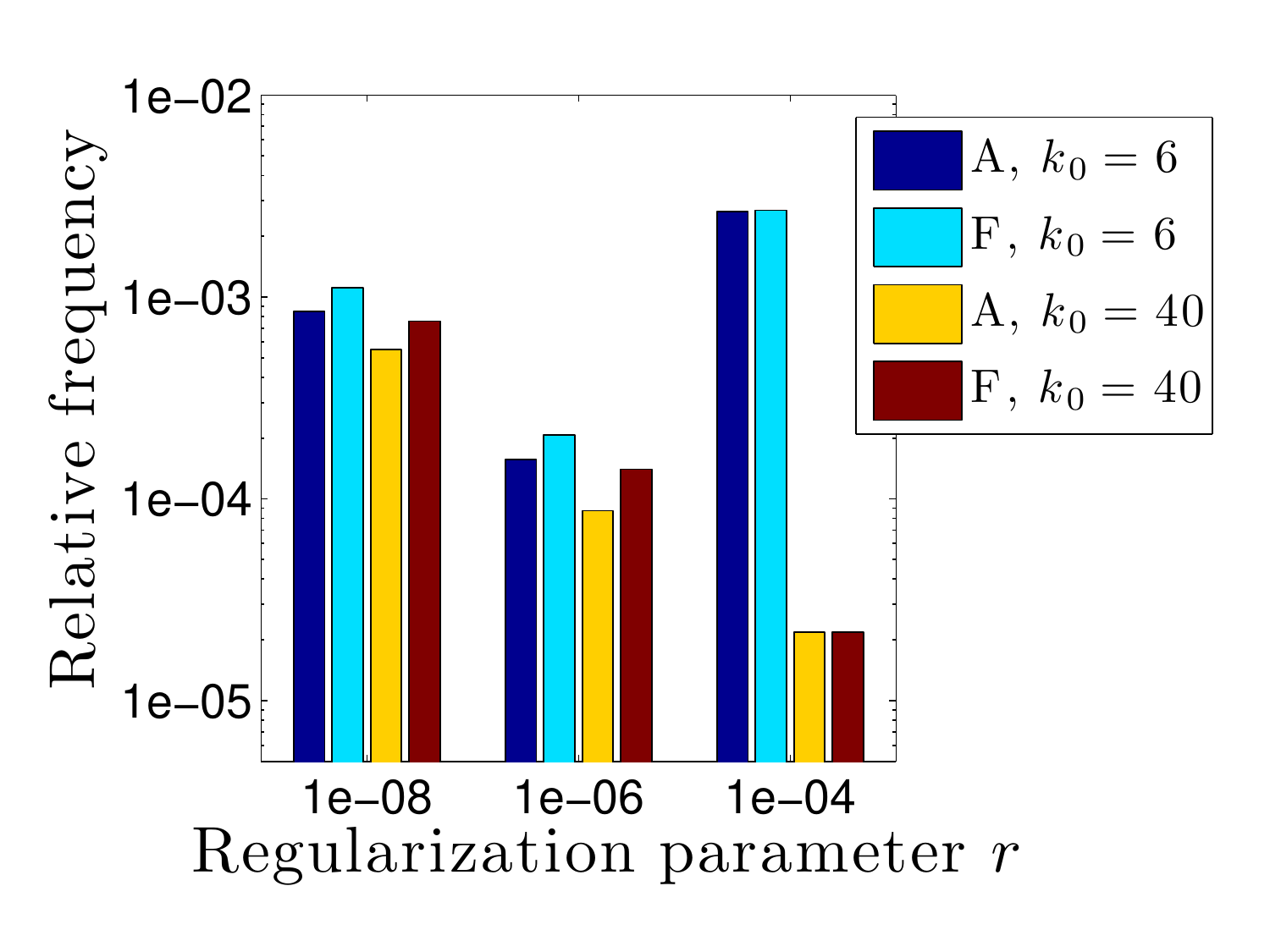}
 \caption{Regularization for the plane source problem.}
 \label{fig:planeN15r-hist}
\end{figure}

\begin{figure}
 \centering
 \subfigure[$k_0 = 40$ vs. $k_0 = 6$ with the adaptive basis.]
  {\label{subfig:planeN15u-diff-c-40-vs-6}
  \includegraphics[width = .4\textwidth]
   {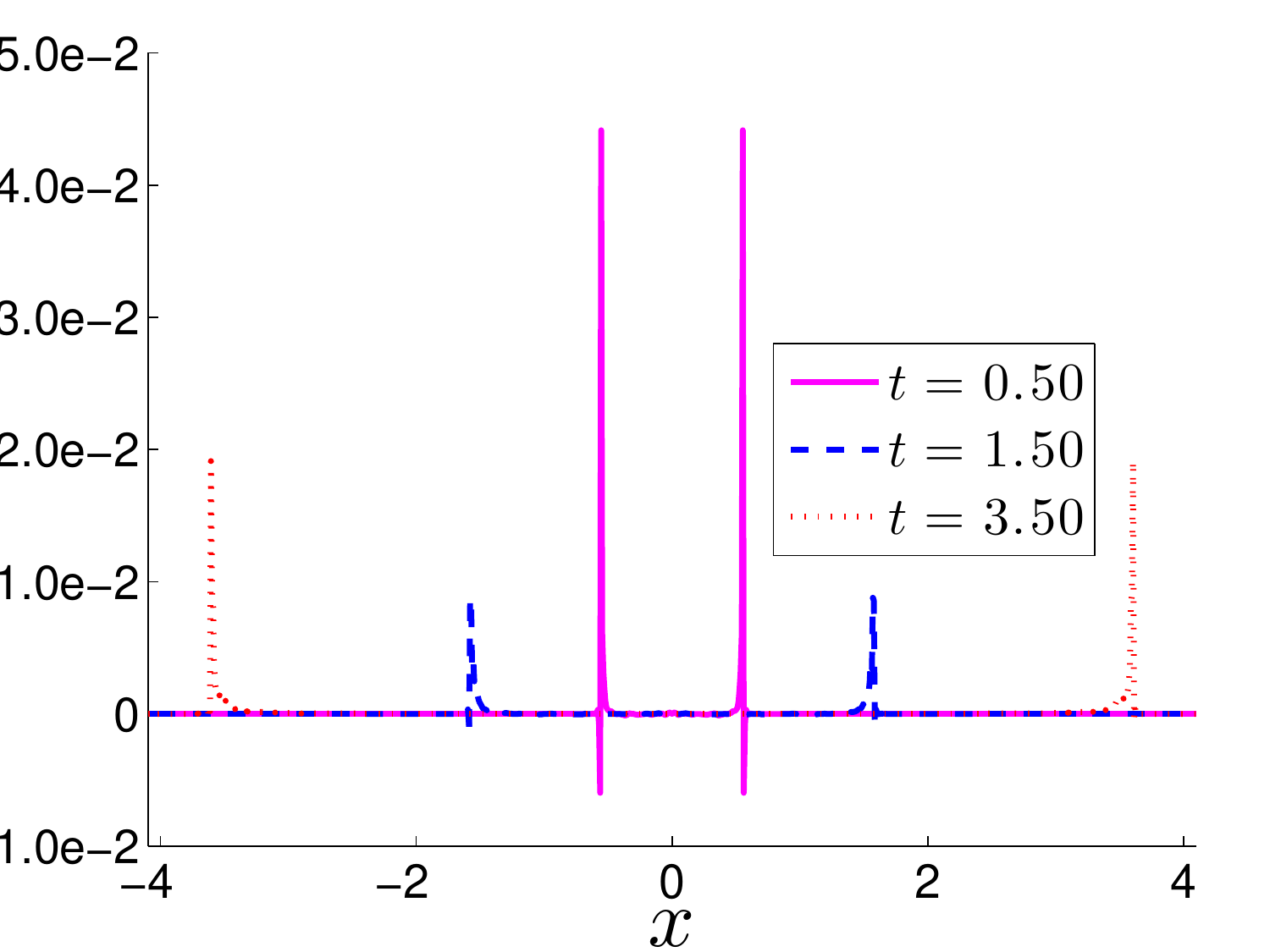}}
 \hspace{.2in}
 \subfigure[Adaptive vs. fixed basis with $k_0 = 40$.  ]
  {\label{subfig:planeN15u-diff-40-c-vs-f}
  \includegraphics[width = .4\textwidth]
   {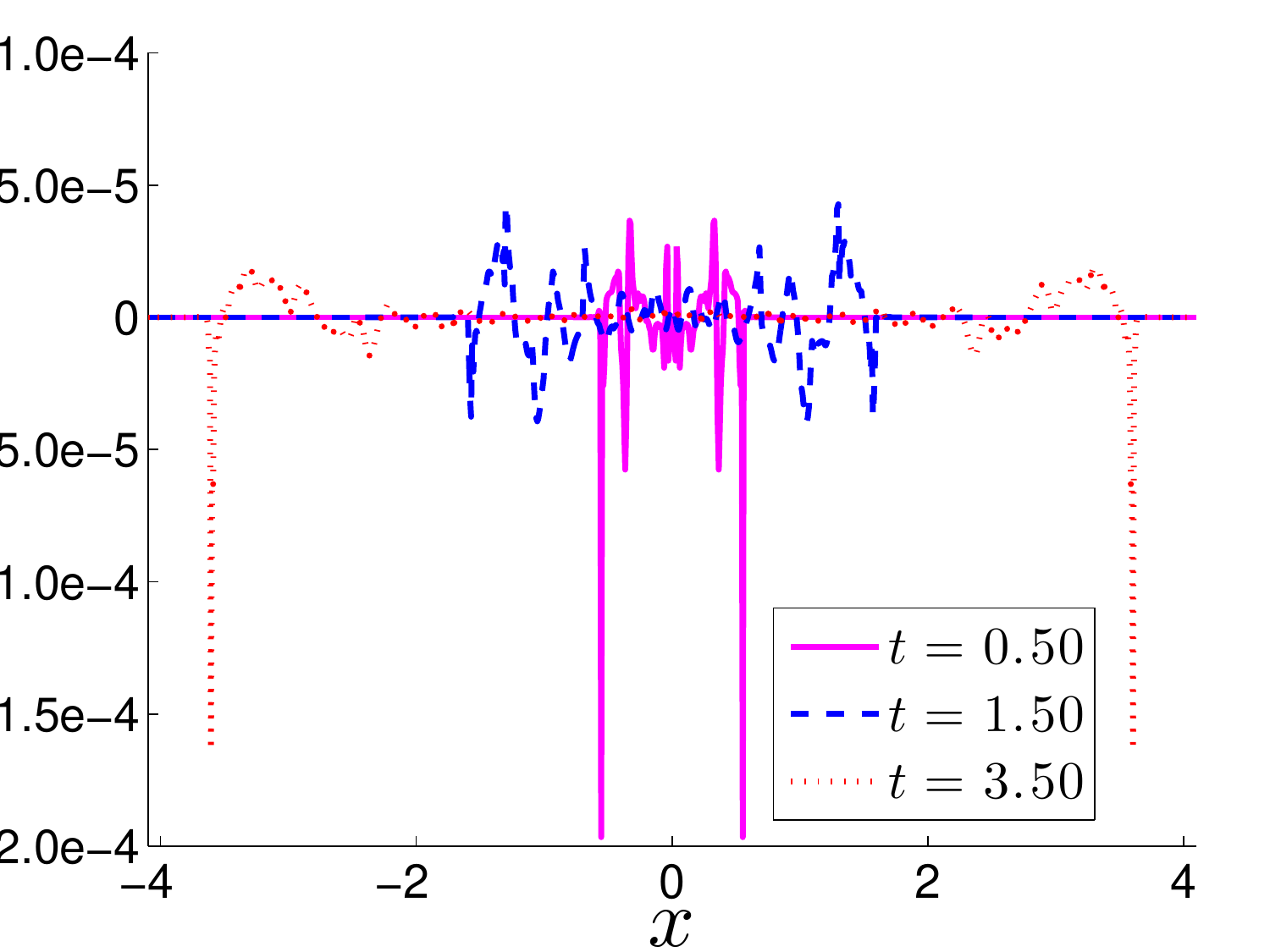}} \\
 \subfigure[$k_0 = 40$ vs. $k_0 = 6$ with the fixed basis.]
  {\label{subfig:planeN15u-diff-f-40-vs-6}
  \includegraphics[width = .4\textwidth]
   {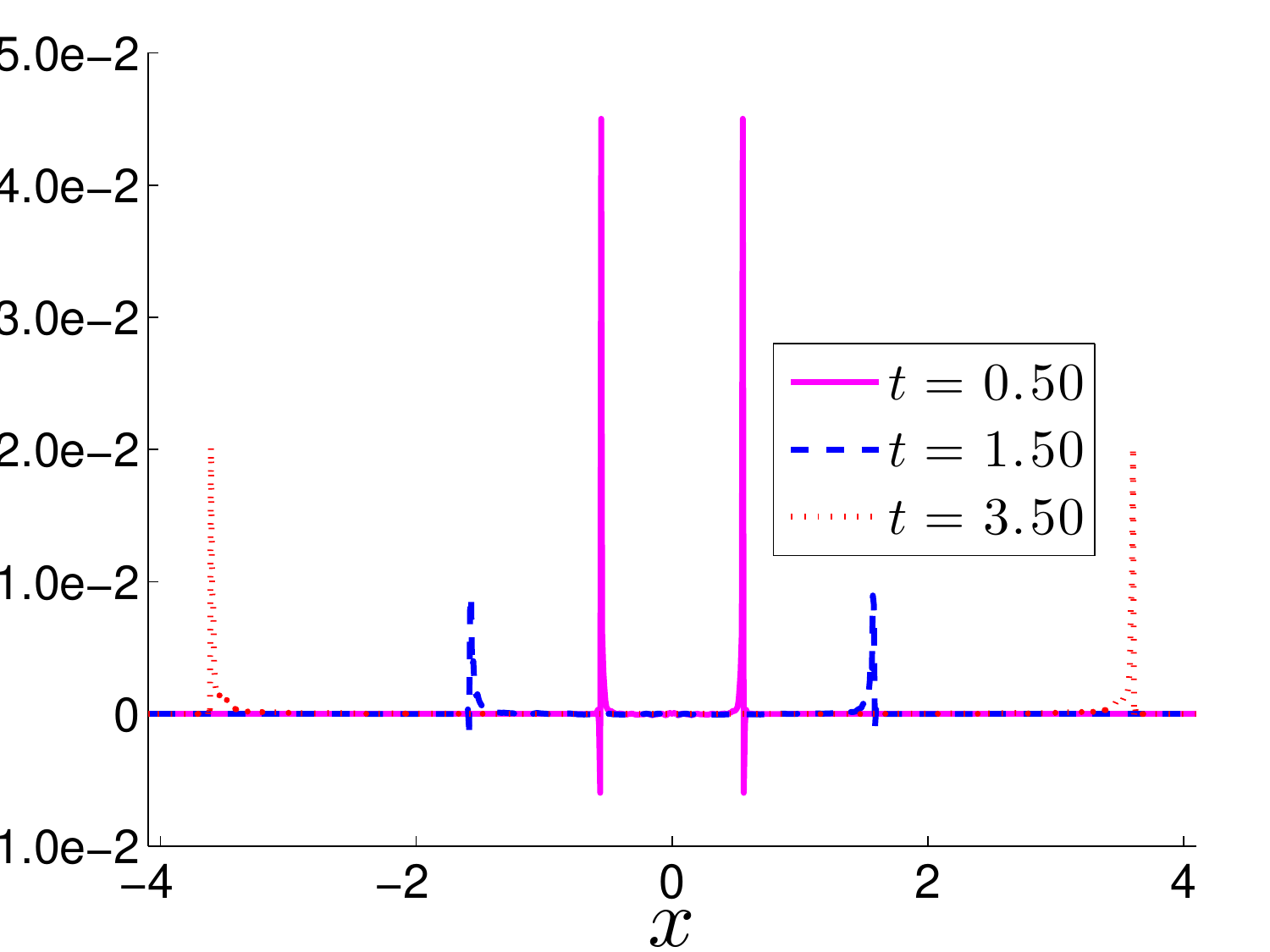}}
 \hspace{.2in}
 \subfigure[Adaptive vs. fixed basis with $k_0 = 6$.  ]
  {\label{subfig:planeN15u-diff-6-c-vs-f}
  \includegraphics[width = .4\textwidth]
   {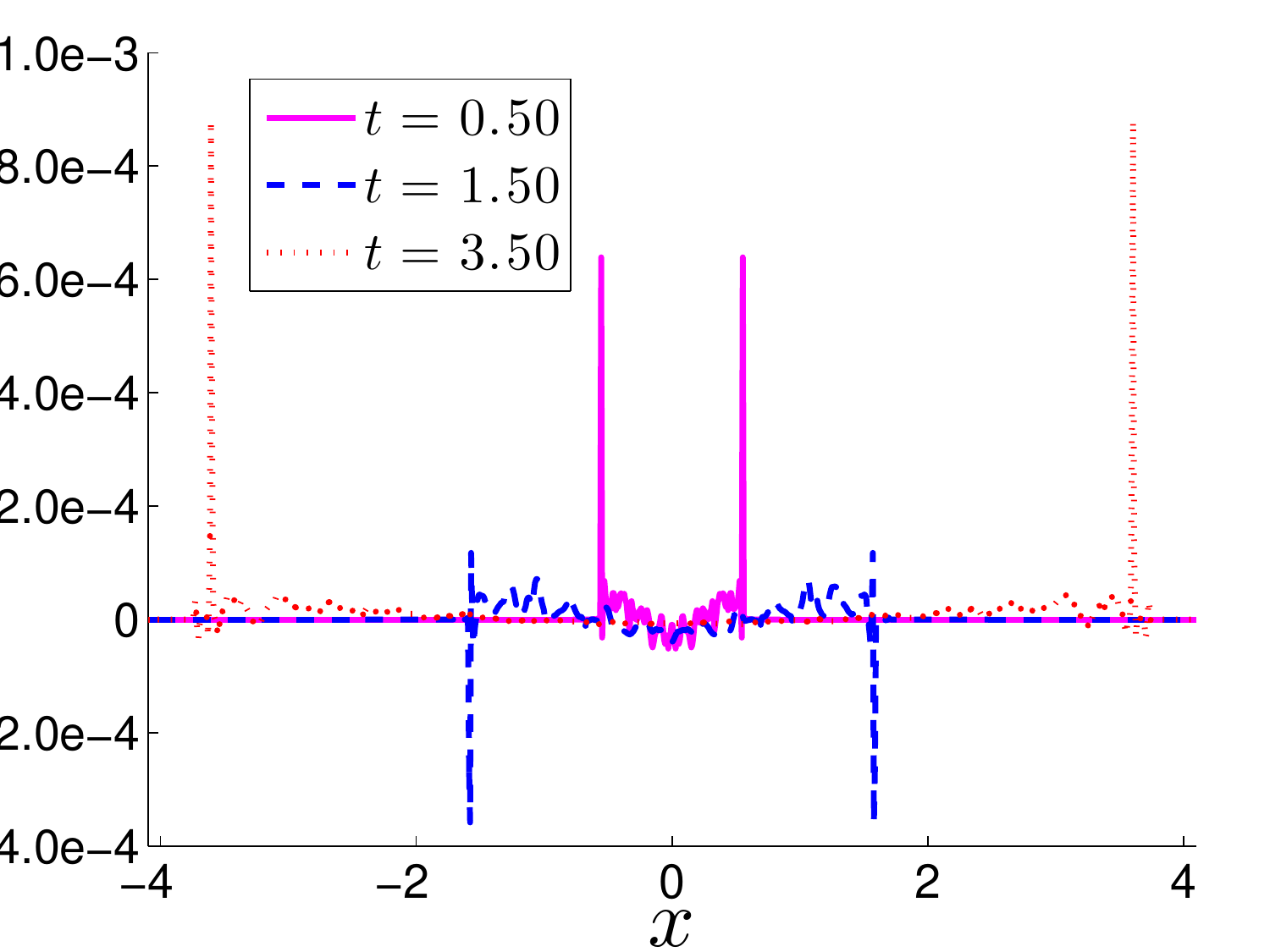}}
 \caption{Relative differences between plane-source solutions. To compare
  $u_0^{(1)}$ vs. $u_0^{(2)}$ we plot $(u_0^{(1)} - u_0^{(2)}) 
  / (0.5(u_0^{(1)} + u_0^{(2)}))$.}
 \label{fig:plane-diff}
\end{figure}

\figref{fig:plane-diff} shows the relative differences between the solutions
using the two different regularization-scheme parameters, $k_0 = 6$ and $k_0 =
40$, with both adaptive- and fixed-basis methods. The figures show that the
biggest differences appear just as the particles enter the surrounding vacuum. 
In each case, the signs of the errors indicate that the solution computed using
less regularization is larger at these points.  This indicates that the solution
with this scheme is advancing slightly faster, though we note that this
difference decreases with time when comparing regularization parameters. When
comparing basis methods, the relative errors are a few orders of magnitude
smaller.

\begin{figure}
 \subfigure[$k_0 = 40$]
  {\label{subfig:planeN15iter-hist-log-k0-40}
  \includegraphics[width = .4\textwidth]
   {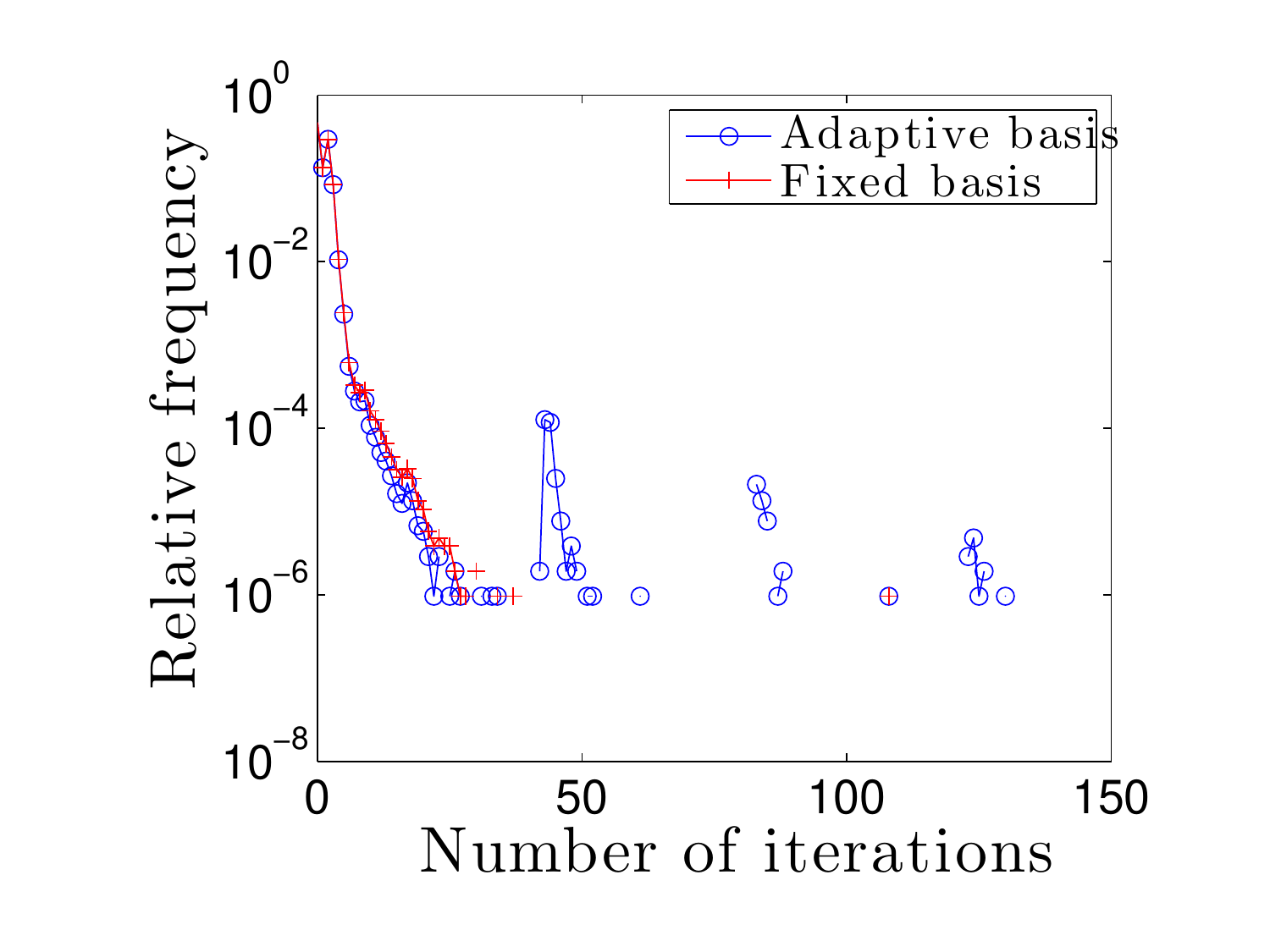}}
 \hspace{.2in}
 \subfigure[$k_0 = 6$]
  {\label{subfig:planeN15iter-hist-log-k0-6}
  \includegraphics[width = .4\textwidth]
   {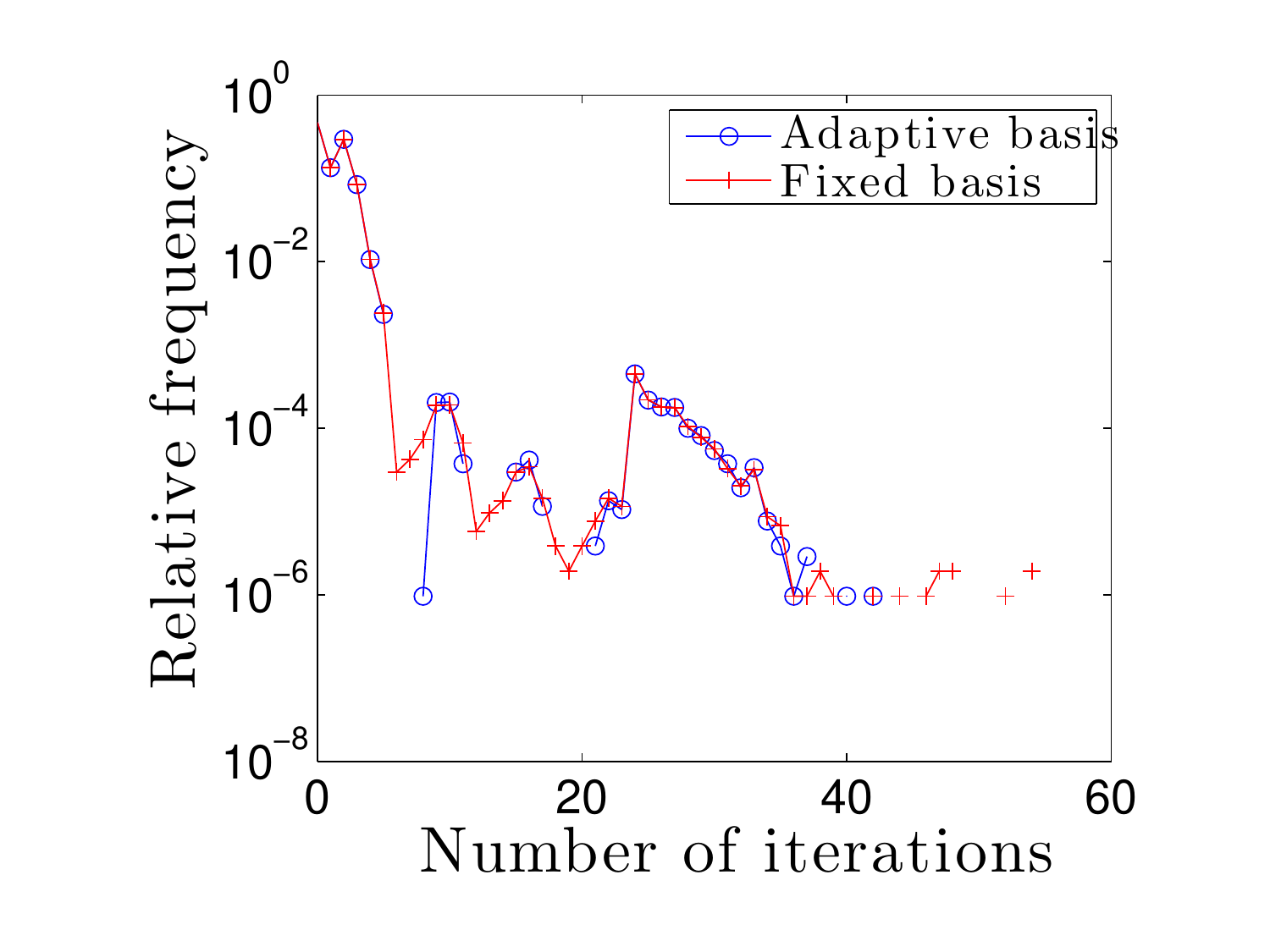}}
 \caption{Comparing iteration histograms between adaptive- and fixed-basis
  methods on the plane-source problem.}
 \label{fig:plane-log-iter-hist}
\end{figure}

In \figref{fig:plane-log-iter-hist} we compare more closely the iteration
histograms of adaptive- and fixed-basis methods. In a parallel implementation,
optimization problems requiring many iterations become a bottleneck, so an
optimizer that needs fewer iterations has a significant advantage.  However,
the results here show that the iteration histograms are nearly identical.  This
is consistent with our numerical experience that when both methods can solve a
problem, they typically take the same number of iterations.  (The figure does
not include a few outlying problems from the adaptive basis: with $k_0 = 40$,
one problem took $529$ iterations and three took $530$ iterations; with $k_0 =
6$, one problem took $446$ iterations and one took $448$ iterations.)

\subsubsection{Two-beam instability}
\label{subsubsec:beams}

In this problem, particles constantly stream into the domain from the left at
$\xL=-0.5$ and the right at $\xR=0.5$ into the initially (almost) vacuous
interior. There is no scattering: $\sig{s}=0$, while $\sig{t}=2$. We use
`forward-peaked' boundary conditions,
\begin{equation}
 \bu(\xL,t) = \Vint{\bm \exp(-10(\mu-1)^2)},
 \quad
 \bu(\xR,t) = \Vint{\bm \exp(-10(\mu+1)^2)} \,.
\end{equation}
On the interior, the initial condition is isotropic with $u_0(x,0) \equiv F_{\rm
floor}\Vint{1}$.

\figref{fig:beamsN15} presents results from a $N_x = 1000$ cell simulation of
the $\M_{15}$ system.  For the this figure we have again used the adaptive-basis
method and $k_0 = 6$, but the corresponding figures using $k_0 = 40$ or the
fixed-basis method are qualitatively similar. The difference between the final
solutions at steady-state using these two different pairs of parameters values
is $O(10^{-6})$ in the $L^\infty$--norm.  The transient profile of the particle
density $u_0(x,t)$ is shown in Figure \ref{subfig:beamsN15u}, where we can also
see that the steady-state is qualitatively indistinguishable from the
steady-state particle density of the kinetic system. These results again agree
qualitatively with what was presented in \cite{AHT10, Hauck-2011}.
\figref{subfig:beamsN15iter} shows the iteration profile.  The mean number of
iterations (excluding trivially solvable isotropic problems and cells which had
already converged) was about $1.37$ for $k_0 = 6$, while for $k_0 = 40$ that
mean was $1.43$. The iteration histogram in Figure
\ref{subfig:beamsN15iter-hist} shows that indeed about $99\%$ of the
optimization problems are solved in three iterations or fewer with $k_0 = 6$.
(With $k_0 = 40$, about $98\%$ are solved in three iterations or fewer.)
Finally, \figref{subfig:beamsN15r} shows that regularization occurred mostly
where particles from the boundary push into the interior vacuum along the front
$x = \pm 0.5 \mp t$ for $t \in [0, 0.5]$.

The histogram in Figure \ref{fig:beamsN15r-hist} shows that the values of the
regularization parameter $r$ used for all nontrivial optimization problems are
similar to those for the plane-source simulations.  For the adaptive basis,
about $0.12\%$ of the problems were regularized with $k_0 = 6$ and about
$0.087\%$ of the problems were regularized with $k_0 = 40$, with smaller values
of $r$. The fixed-basis method, on the other hand, used significantly more
regularization for both values of $k_0$. Again, the amount of regularization in
all cases is significantly less than in \cite{AHT10}, where roughly 1.3\% of the
problems were regularized.

\begin{figure}
 \centering
 \subfigure[Snapshots of the solution, $u_0(x,t)$ at $t=0.2$, $0.4$, $0.8$, and
  $4$.  A green curve shows the true steady-state solution, which is
  qualitatively indistinguishable from the numerical moment solution at $t =
  4$.]
  {\label{subfig:beamsN15u}
  \includegraphics[width = .4\textwidth]
  {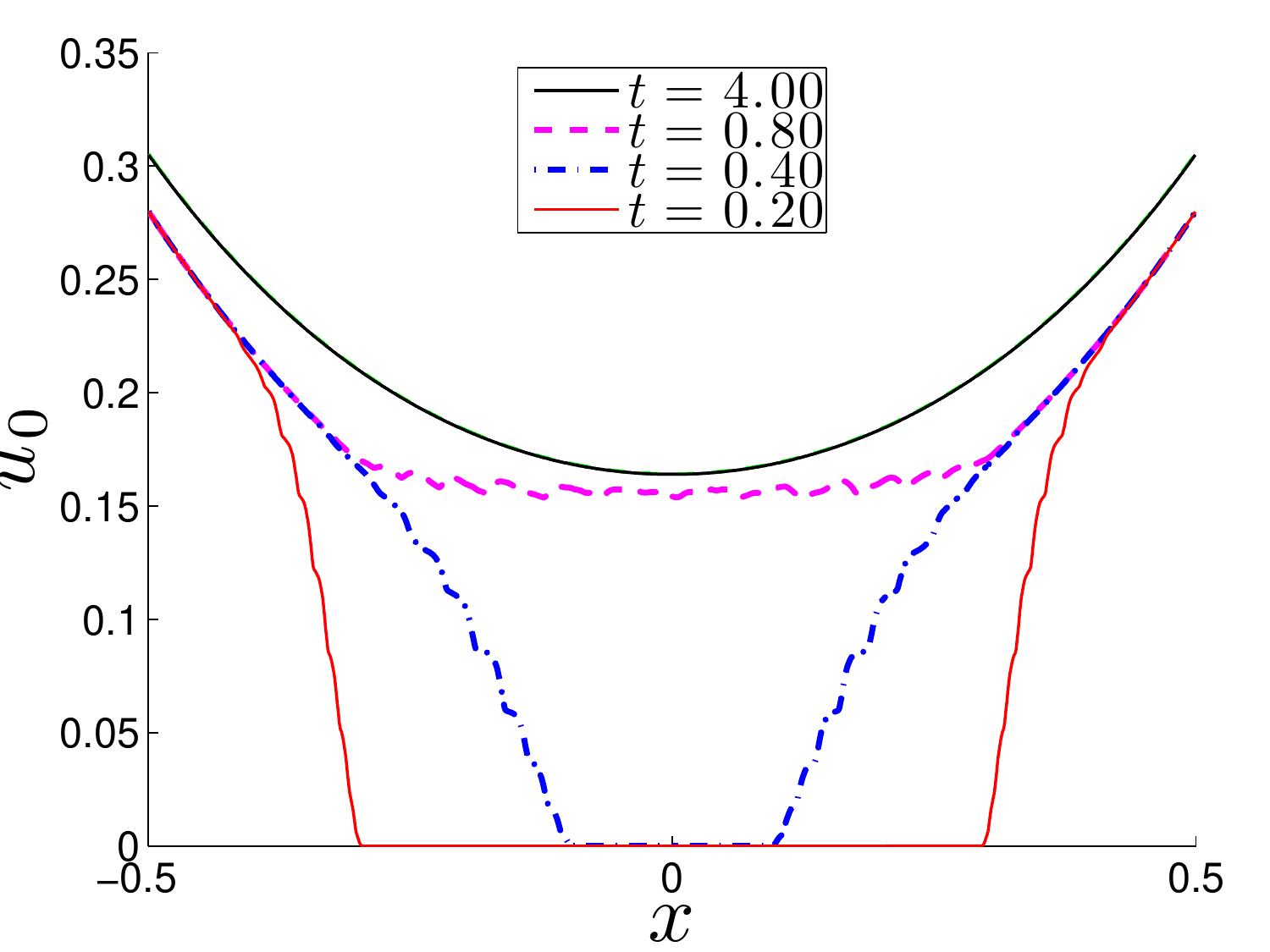}}
 \hspace{.2in}
 \subfigure[The total number of iterations (over the two Runge-Kutta stages) 
  needed to solve the optimization problem at each point in space and time. 
  The maximum number of iterations needed for one time step was $70$ (off
  scale).]{\label{subfig:beamsN15iter}
  \includegraphics[width = .4\textwidth]
   {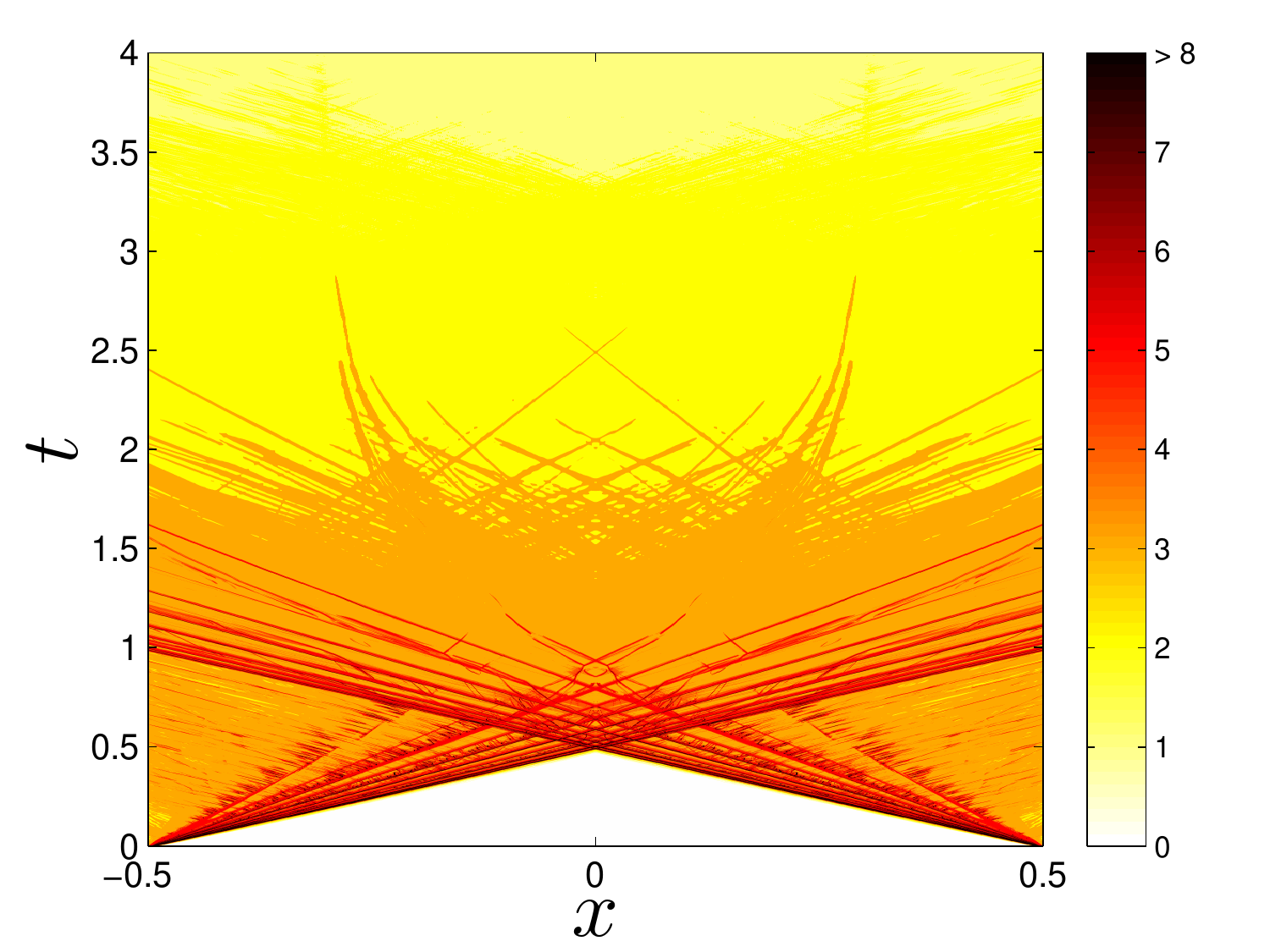}} \\
 \vspace{0.5in}
 \subfigure[A histogram of number of iterations needed to solve each
  optimization problem.  About $0.13\%$ of the nontrivial problems
  needed more than six iterations.] {\label{subfig:beamsN15iter-hist}
  \includegraphics[width = .4\textwidth]
   {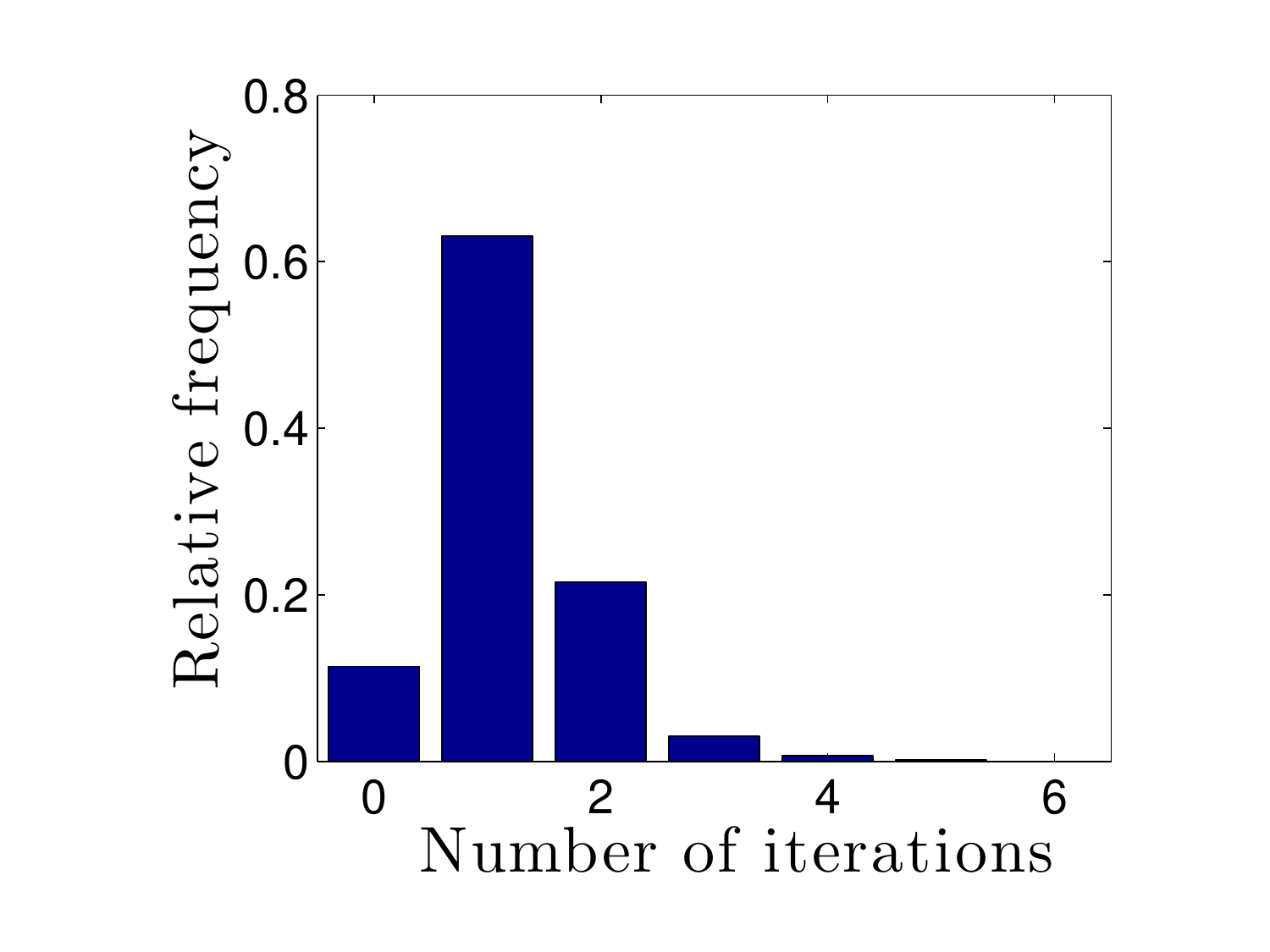}}
 \hspace{.2in}
 \subfigure[The locations of regularizations.]{\label{subfig:beamsN15r}
  \includegraphics[width = .4\textwidth]
   {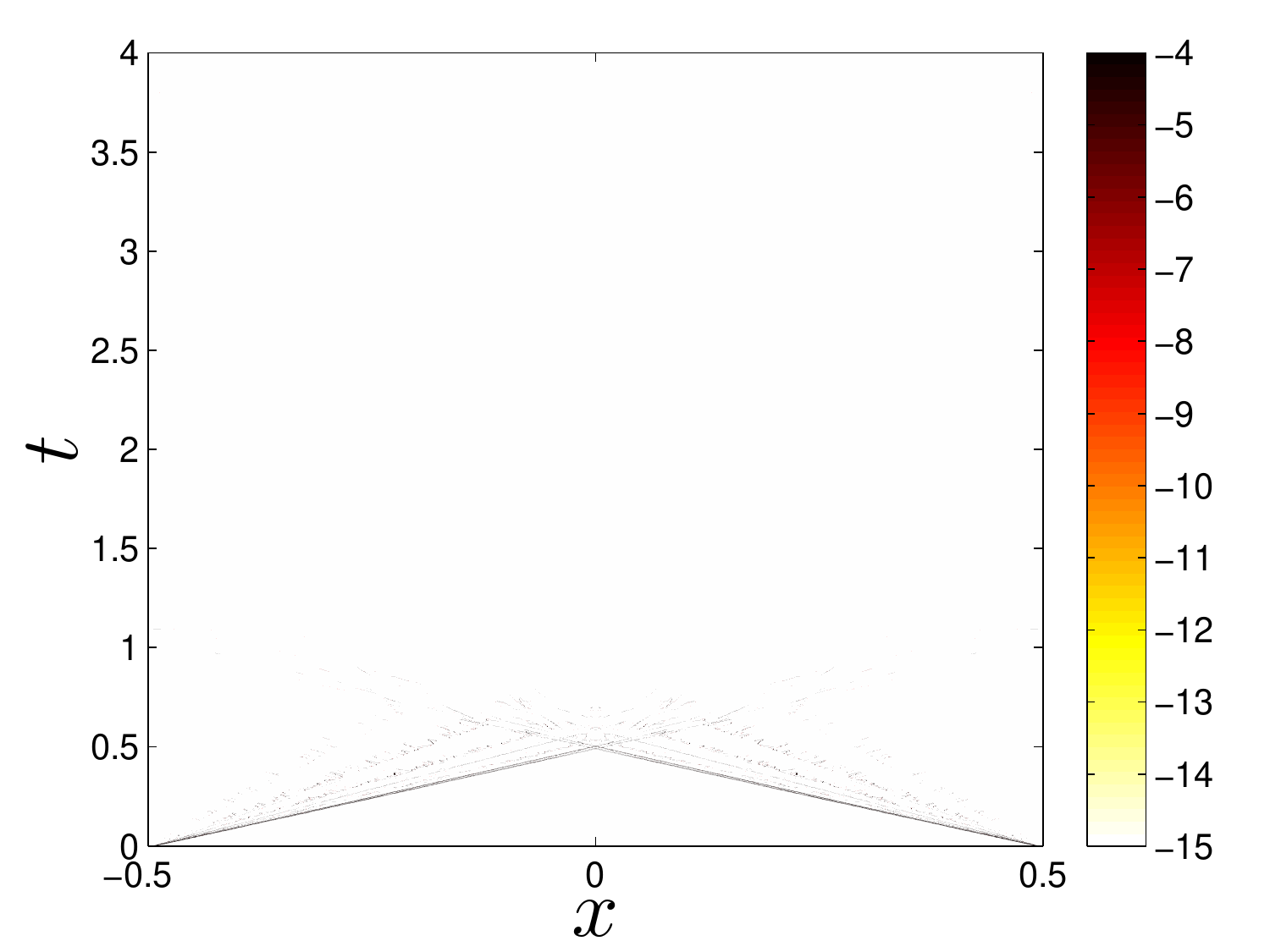}}
 \caption{Results for the $\M_{15}$ model of the two-beam instability with
  $N_x = 1000$ cells for $k_0 = 6$.  For this simulation,
  excluding trivial cells, the optimization problem is solved about $15
  \times 10^6$ times.}
 \label{fig:beamsN15}
\end{figure}

\begin{figure}
 \centering
 \includegraphics[width = 0.6\textwidth]{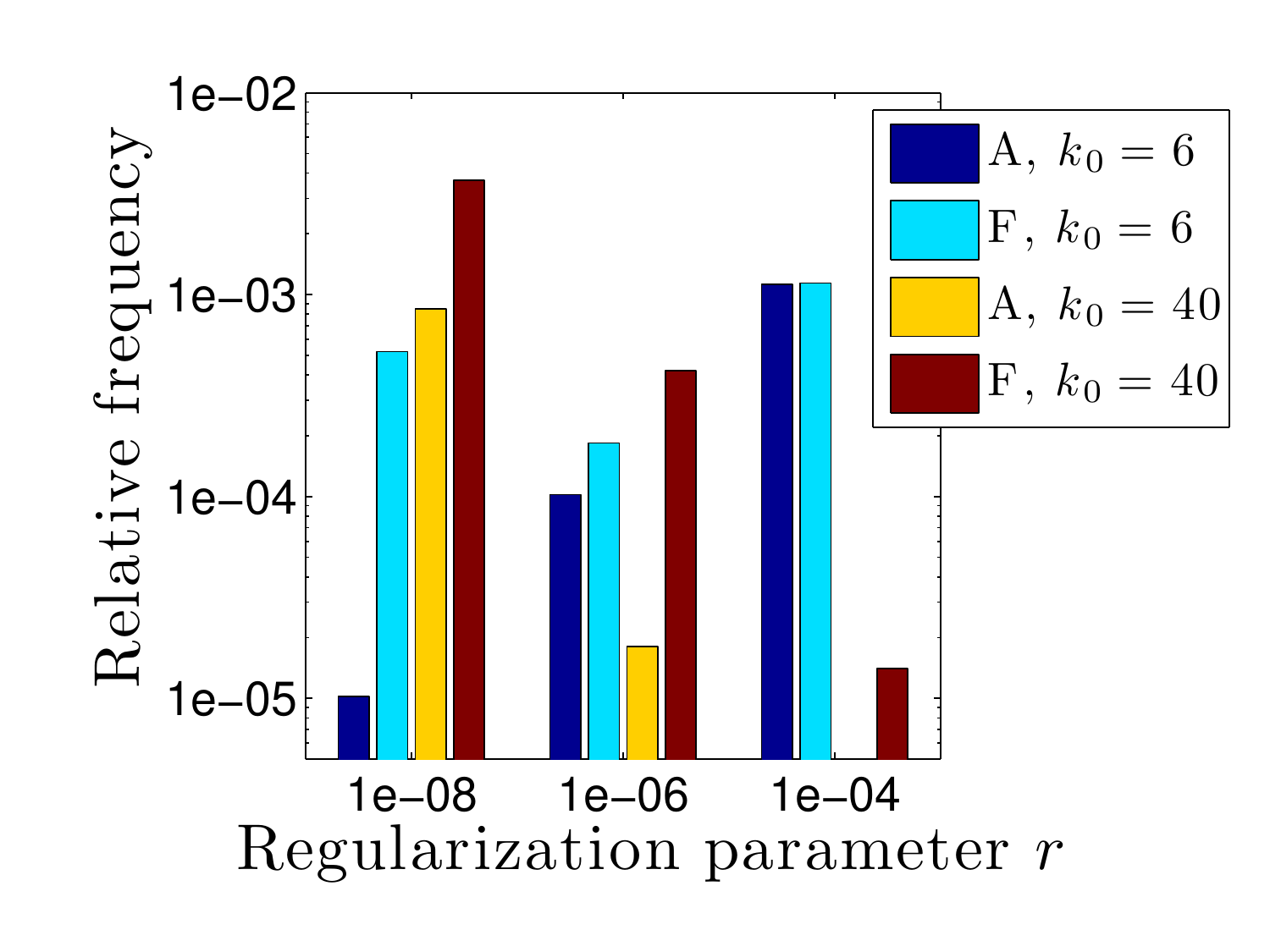}
 \caption{Regularization for the two-beam problem.}
 \label{fig:beamsN15r-hist}
\end{figure}

\begin{figure}
 \centering
 \subfigure[$k_0 = 40$ vs. $k_0 = 6$ with the adaptive basis.  Off-scale peaks
  at $t = \pm 0.2$ are $O(1)$.]
  {\label{subfig:beamsN15u-diff-c-40-vs-6}
  \includegraphics[width = .4\textwidth]
   {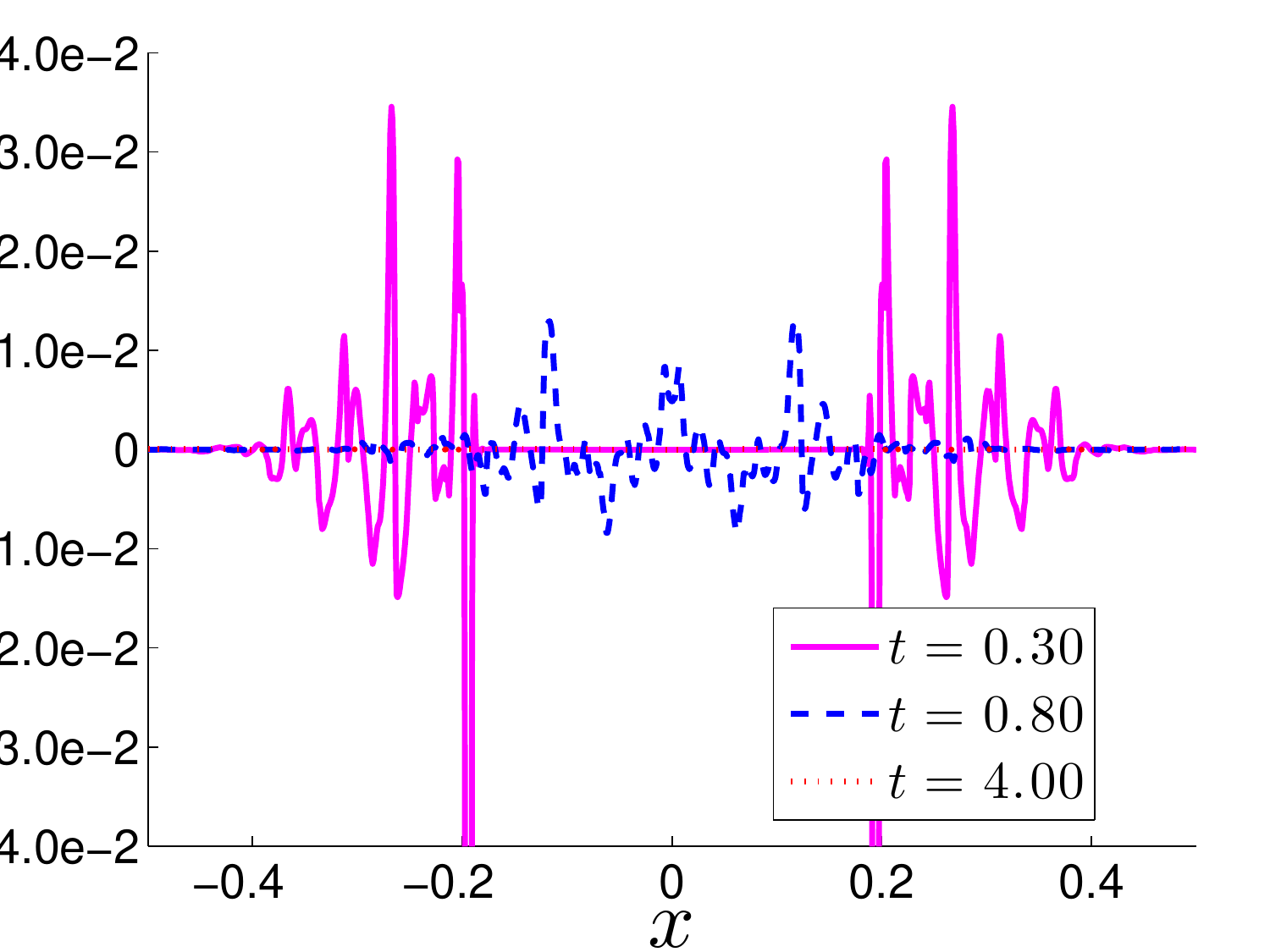}}
 \hspace{.2in}
 \subfigure[Adaptive vs. fixed basis with $k_0 = 40$.  Off-scale peaks
  at $t = \pm 0.2$ are $O(100)$.]
  {\label{subfig:beamsN15u-diff-40-c-vs-f}
  \includegraphics[width = .4\textwidth]
   {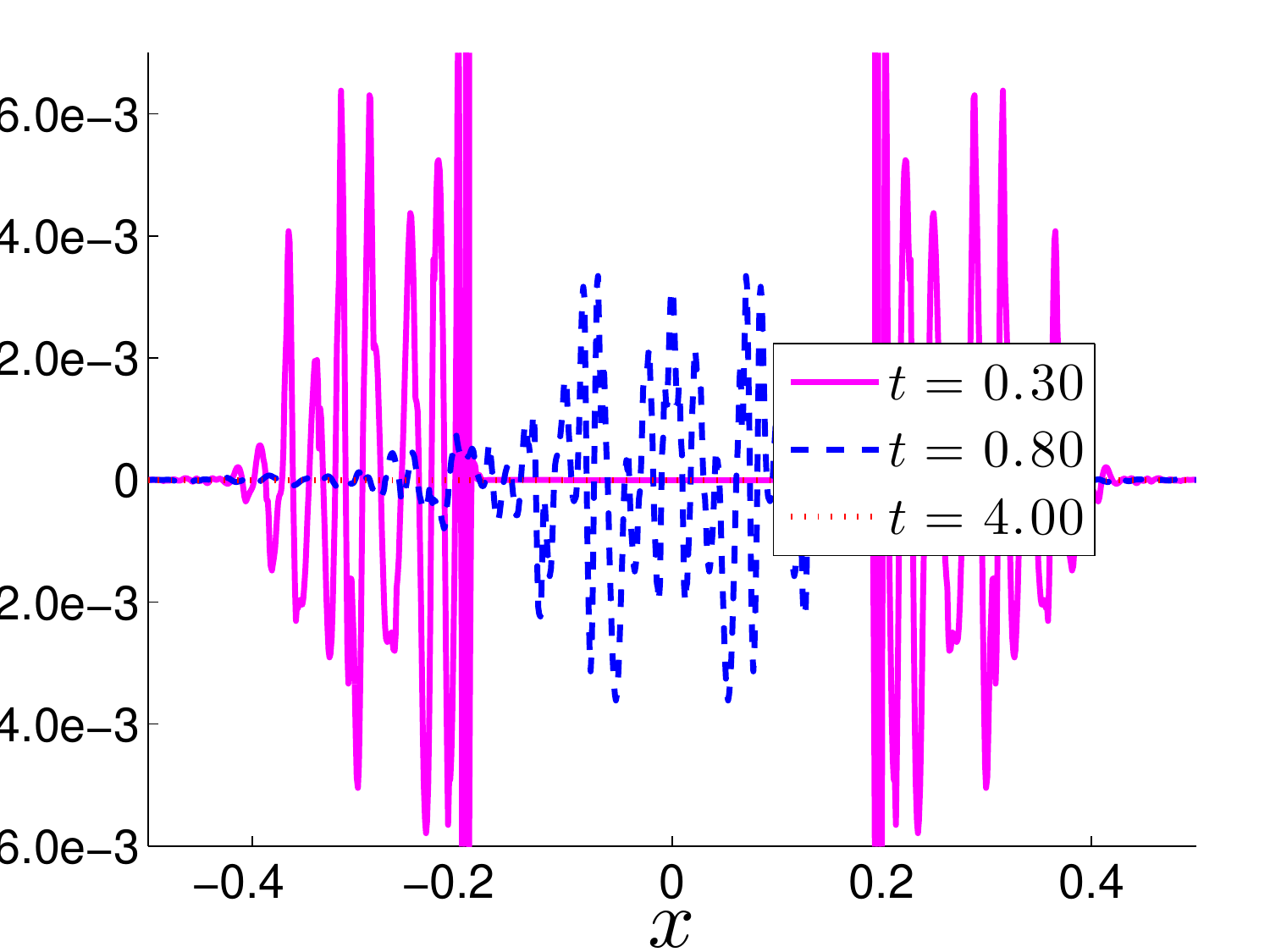}} \\
 \subfigure[$k_0 = 40$ vs. $k_0 = 6$ with the fixed basis .  Off-scale peaks
  at $t = \pm 0.2$ are $O(1)$.]
  {\label{subfig:beamsN15u-diff-f-40-vs-6}
  \includegraphics[width = .4\textwidth]
   {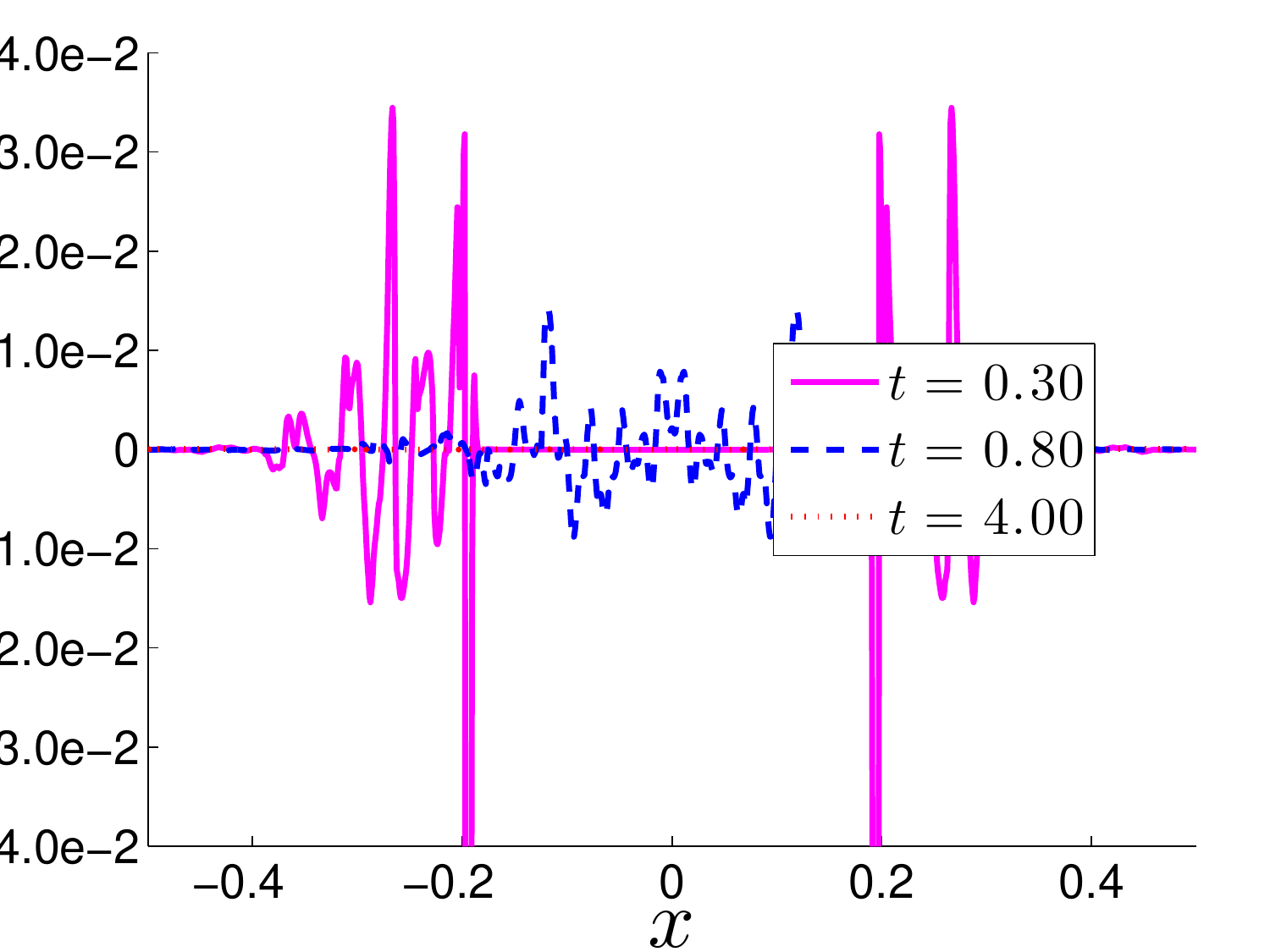}}
 \hspace{.2in}
 \subfigure[Adaptive vs. fixed basis with $k_0 = 6$.  Off-scale peaks
  at $t = \pm 0.2$ are $O(0.1)$.]
  {\label{subfig:beamsN15u-diff-6-c-vs-f}
  \includegraphics[width = .4\textwidth]
   {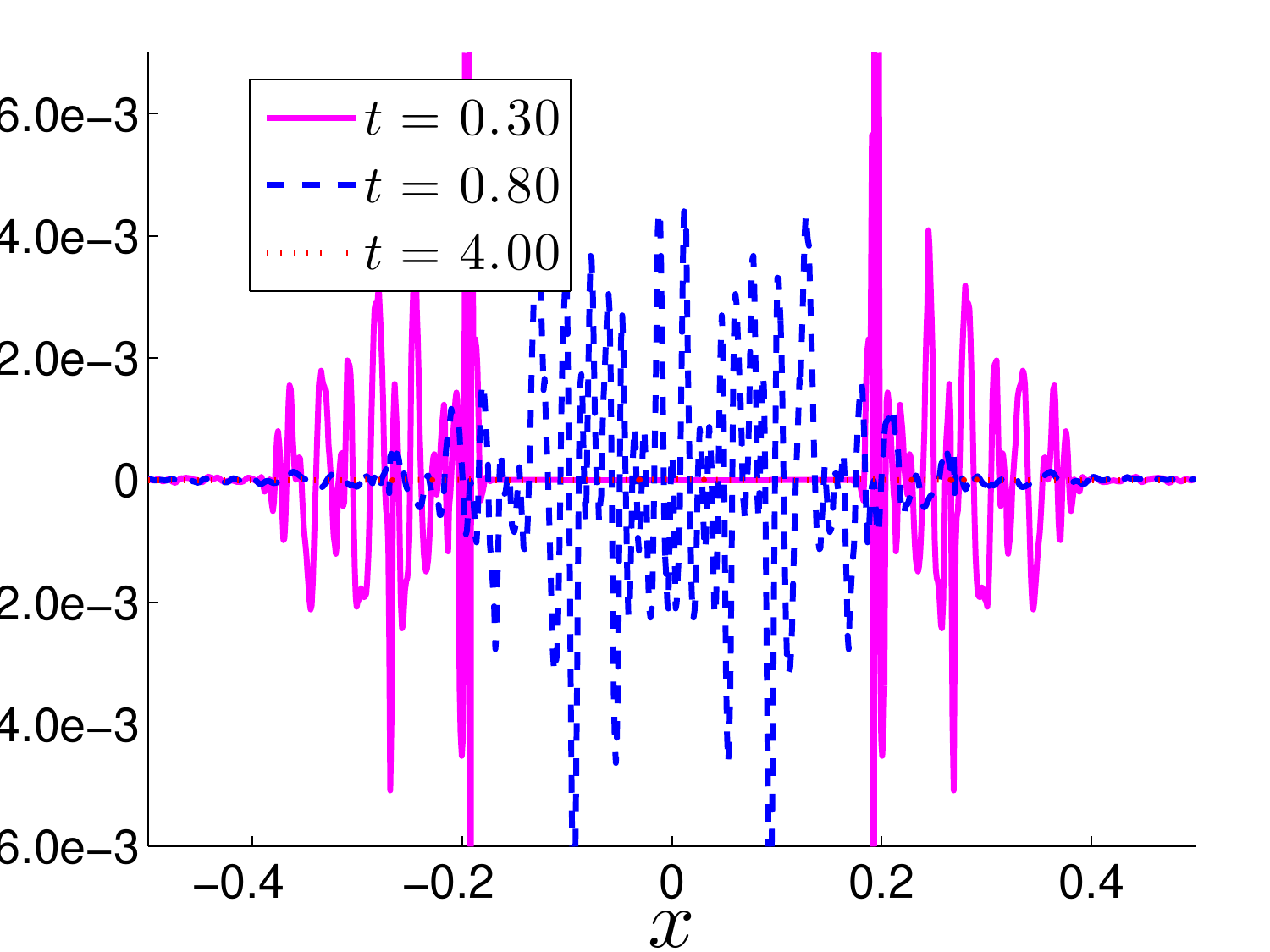}}
 \caption{Relative differences between two-beam solutions: let $u_0^{(1)}$
  and $u_0^{(2)}$ represent the particle densities for the solutions from two
  different methods. Above we plot $(u_0^{(1)} -
  u_0^{(2)})/(0.5(u_0^{(1)} + u_0^{(2)}))$.}
 \label{fig:beams-diff}
\end{figure}

\figref{fig:beams-diff} shows the relative difference in the solutions from the
two different regularization schemes with both basis methods.  The results here
are not as easily interpreted as in the plane source problem.  For $t = 0.3$,
the sign of the largest two peaks indicate that the solution with more
aggressive regularization has advanced more quickly, but again for later time,
these differences decrease.

\begin{figure}
 \subfigure[$k_0 = 40$]
  {\label{subfig:beamsN15iter-hist-log-k0-40}
  \includegraphics[width = .4\textwidth]
   {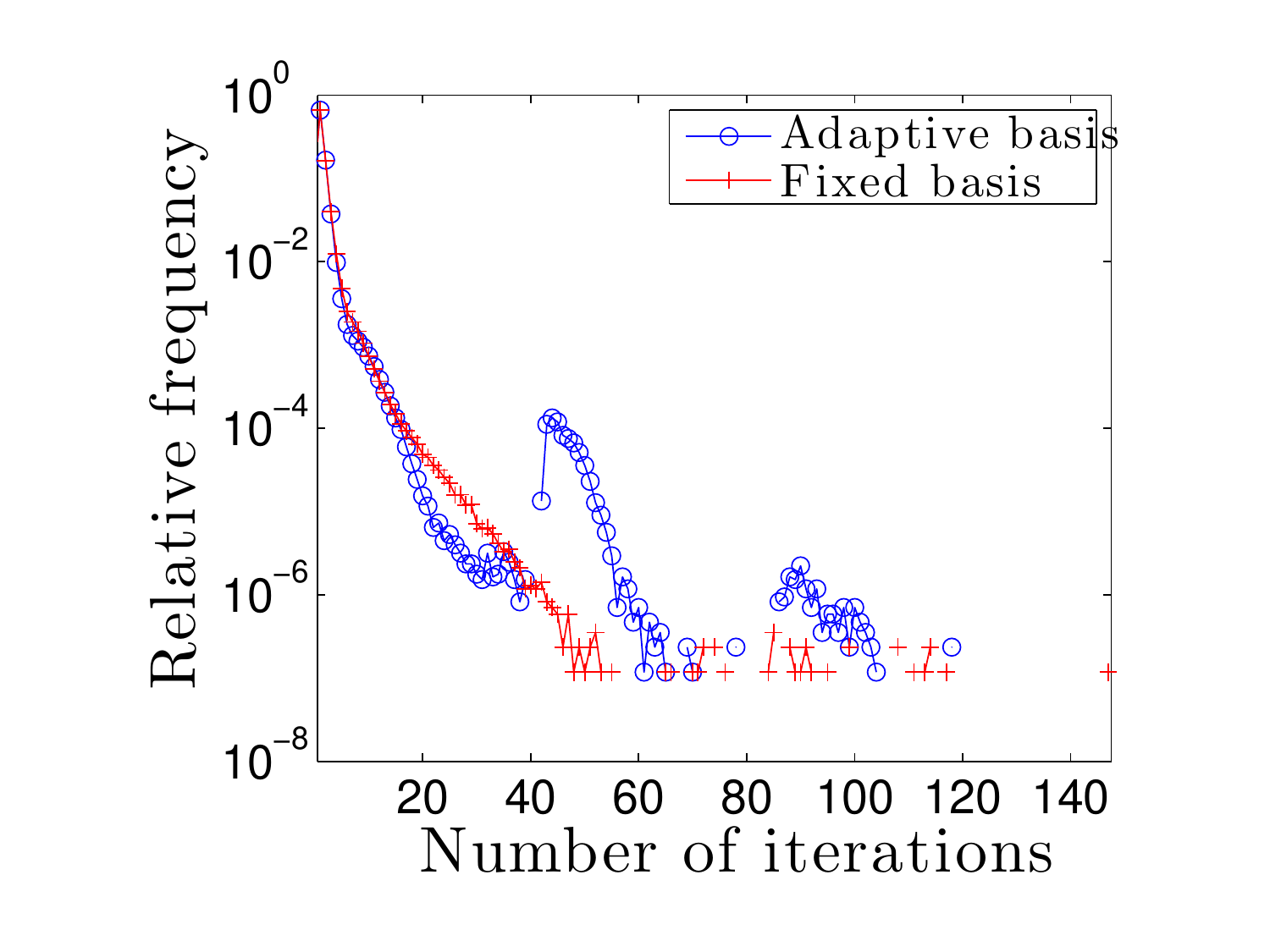}}
 \hspace{.2in}
 \subfigure[$k_0 = 6$]
  {\label{subfig:beamsN15iter-hist-log-k0-6}
  \includegraphics[width = .4\textwidth]
   {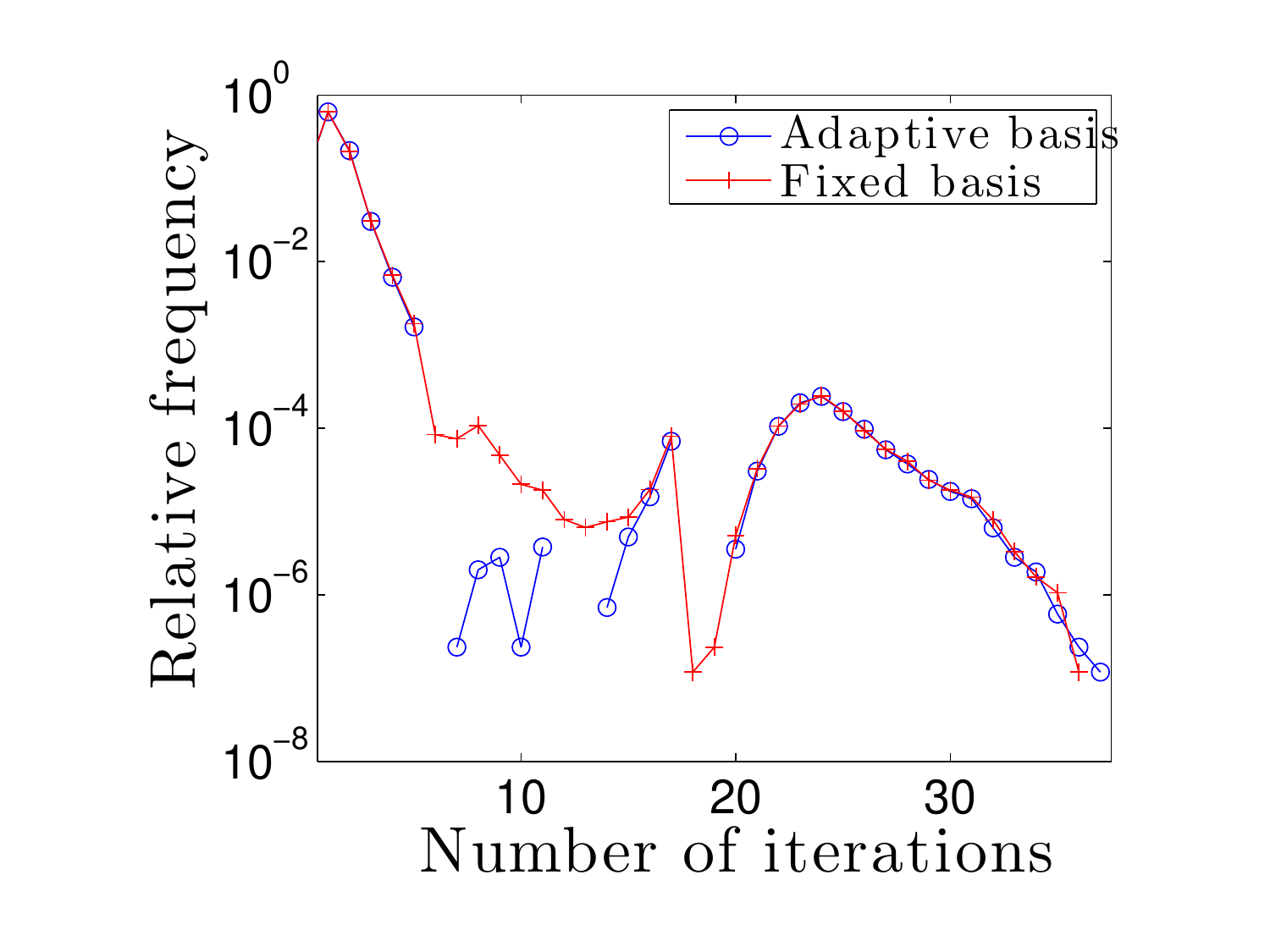}}
 \caption{Comparing iteration histograms between adaptive- and fixed-basis
  methods on the two-beam problem.}
 \label{fig:beams-log-iter-hist}
\end{figure}

In \figref{fig:beams-log-iter-hist} we again compare more closely the iteration
histograms of adaptive- and fixed-basis methods.  Unlike in the plane-source
problem, here with $k_0 = 40$ the adaptive-basis optimizer uses a high number of
iterations on significantly more problems.  This is clearly a consequence of our
fairly rudimentary regularization scheme: indeed, notice the increases in
histogram near $40$ and $80$ iterations for the adaptive-basis method in
\figref{subfig:beamsN15iter-hist-log-k0-40}.  They exactly lie where the
regularization scheme increases $r$.  On the other hand, with the fixed-basis
method, $r$ is often increased before $k_0$ iterations have passed because the
condition number of the Hessian becomes unacceptably high.  With adaptive-basis,
the condition number of the Hessian stays close to one.

\section{Conclusions}
\label{sec:conclusions}

We have presented a complete and practical numerical algorithm for solving the
$\M_N$ entropy-based moment closure model in slab geometry. For the optimization
at each space-time grid point, our method uses a change of polynomial basis to
keep the Hessian matrix near the identity. This method is closely related to
that presented in \cite{Abramov-2009}, although we have used
the Cholesky factorization instead of the Gram-Schmidt method to define
the change of basis. This former is more efficient and, in a series of tests, has
been shown to perform comparably.
We tested our method on challenging test problems including
a new set of manufactured solutions as well as the standard plane source and
two-beam problems. Numerical results indicate that the new method has many
advantages over the use of a fixed basis such as Legendre polynomials. First,
the adaptive basis allows solution of optimization problems closer to the
boundary of realizability. This leads to a decreased use of regularization.  We
show, using the manufactured solution, that regularization introduces errors
which can significantly slow or even stall convergence.  Therefore, the
decreased use of regularization observed with the adaptive-basis method leads to
improved accuracy. Finally we show that the adaptive-basis method performs
better with low-resolution quadrature than a fixed-basis method with high
resolution quadrature.  Thus we can avoid the use of adaptive quadrature (an
alternative to improve conditioning of the Hessian), and by using a fixed
quadrature the numerical solution remains in a constant computational region of
realizability.

Of course, the one-dimensional kinetic model here is simply a testbed.  Real
problems of practical interest are in two and three dimensions in space
and velocity, so future work should test these methods there.  A parallel
implementation is also necessary to fully exploit the advantages of
entropy-based moment closures while minimizing the computational burden of the
numerical optimization.

\bibliographystyle{model1-num-names}
\bibliography{entropy_kinetic}

\end{document}